\documentclass[11pt,english]{article}

\usepackage[margin=1in]{geometry}
\usepackage{bbm}
\usepackage{graphicx,color}
\usepackage{amsmath,amssymb,amsthm}
\usepackage{enumerate}
\usepackage{enumitem}
\usepackage{fullpage}
\usepackage{algorithm}
\usepackage{mathrsfs}
\usepackage[noend]{algcompatible}
\usepackage[noblocks]{authblk}
\usepackage{varwidth}
\usepackage{mathrsfs}
\usepackage{babel}
\usepackage{accents}
\usepackage{mathtools}
\usepackage{tipa}
\usepackage[most]{tcolorbox}
\usepackage{wrapfig}
\usepackage{blindtext}
\usepackage{blindtext}
\usepackage{thm-restate}
\usepackage{longtable}
\usepackage[backref=page, pagebackref=true]{hyperref}
\usepackage[capitalize]{cleveref} 
\Crefname{enumtry}{try1}{}   

\usepackage{booktabs} 

\newtheorem{question}{Question}

\def\denseformat{
\setlength{\textheight}{9in}
\setlength{\textwidth}{6.9in}
\setlength{\evensidemargin}{-0.2in}
\setlength{\oddsidemargin}{-0.2in}
\setlength{\headsep}{10pt}
\setlength{\topmargin}{-0.3in}
\setlength{\columnsep}{0.375in}
\setlength{\itemsep}{0pt}
}

\newtheorem{theorem}{Theorem}[section]

\newtheorem{definition}[theorem]{Definition}
\newtheorem{claim}[theorem]{Claim}
\newtheorem{lemma}[theorem]{Lemma}

\newtheorem{observation}[theorem]{Observation}

\def\boldhead#1:{\par\vskip 7pt\noindent{\bf #1:}\hskip 10pt}
\def\ithead#1:{\par\vskip 7pt\noindent{\it #1:}\hskip 10pt}

\def\inline#1:{\par\vskip 7pt\noindent{\bf #1:}\hskip 10pt}
\def\midinline#1:{\par\noindent{\bf #1:}\hskip 10pt}
\def\dnsinline#1:{\par\vskip -7pt\noindent{\bf #1:}\hskip 10pt}
\def\ddnsinline#1:{\newline{\bf #1:}\hskip 10pt}
\def\largeinline#1:{\par\vskip 7pt\noindent{\large\bf #1:}\hskip 10pt}
\long\def\commhide #1\commhideend{}
\long\def\commfull #1\commend{#1}
\long\def\commabs #1\commenda{}
\long\def\commtim #1\commendt{#1}
\long\def\commb #1\commbend{}
\long\def\commedit #1\commeditend{} 

\long\def\commB #1\commBend{}       

\long\def\commex #1\commexend{}     

\long\def\commsiena #1\commsienaend{}  

\long\def\commBI #1\commBIend{}  

\long\def\CProof #1\CQED{}
\def\qed{\mbox{}\hfill $\Box$\\}
\def\blackslug{\hbox{\hskip 1pt \vrule width 4pt height 8pt
    depth 1.5pt \hskip 1pt}}
\def\QED{\quad\blackslug\lower 8.5pt\null\par}

\long\def\PPP#1{\noindent{\bf Proof:}{ #1}{\quad\blackslug\lower 8.5pt\null}}

\long\def\denspar #1\densend
{#1}

\newif\ifnotesw\noteswtrue
   {\ifnotesw\marginpar[\hfill\(\top\)]{\(\top\)}\fi}%
   {\ifnotesw\marginpar[\hfill\(\bot\)]{\(\bot\)}\fi}

\newcommand{\mnote}[1]%
    {\ifnotesw\marginpar%
        [{\scriptsize\it\begin{minipage}[t]{\marginparwidth}
        \raggedleft#1%
                        \end{minipage}}]%
        {\scriptsize\it\begin{minipage}[t]{\marginparwidth}
        \raggedright#1%
                        \end{minipage}}%
    \fi}

\def\MathF{\hbox{\rm I\kern-2pt F}}
\def\MathP{\hbox{\rm I\kern-2pt P}}
\def\MathR{\hbox{\rm I\kern-2pt R}}
\def\MathZ{\hbox{\sf Z\kern-4pt Z}}
\def\MathN{\hbox{\rm I\kern-2pt I\kern-3.1pt N}}
\def\MathC{\hbox{\rm \kern0.7pt\raise0.8pt\hbox{\footnotesize I}
\kern-4.2pt C}}
\def\MathQ{\hbox{\rm I\kern-6pt Q}}

\newsavebox{\ttop}\newsavebox{\bbot}

\def\mod{\pmod}

\def\eps{\epsilon}
\def\epsi{\varepsilon}

\def\nin{{~\not \in~}}

\newcommand{\mst}{\texttt{MST}}
\newcommand{\msttilde}{\widetilde{\mst}}

\newcommand{\poly}{\mathsf{poly}}

\newcommand{\Ftilde}{\widetilde{F}}
\newcommand{\Ttilde}{\widetilde{T}}
\newcommand{\Ptilde}{\widetilde{P}}
\newcommand{\Qtilde}{\widetilde{Q}}

\newcommand{\wrdram}{\textsf{Transdichotomous~}}
\newcommand{\rdrm}{\textsf{Word RAM~}}
\newcommand{\djset}{\textsc{Union-Find}}
\newcommand{\union}{\textsc{Union}}
\newcommand{\find}{\textsc{Find}}
\newcommand{\link}{\textsc{Link}}

\newcommand{\dm}{\mathsf{Dm}}
\newcommand{\calmst}{\mathcal{MST}}

\newcommand{\light}{\text{light}}
\newcommand{\heavy}{\text{heavy}}
\newcommand{\prune}{\mathsf{pruned}}
\newcommand{\adm}{\mathsf{Adm}}
\newcommand{\high}{\mathsf{high}}
\newcommand{\lowp}{\mathsf{low}^+}
\newcommand{\lowm}{\mathsf{low}^-}
\newcommand{\aug}{\mathsf{aug}}
\newcommand{\srt}{\mathsf{SORT}}
\newcommand{\internal}{\mathsf{intrnl}}
\newcommand{\prefix}{\mathsf{pref}}

\newcommand{\ma}{\mathcal{A}}
\newcommand{\mx}{\mathcal{X}}
\newcommand{\my}{\mathcal{Y}}
\newcommand{\mz}{\mathcal{Z}}
\newcommand{\mf}{\mathcal{F}}
\newcommand{\mv}{\mathcal{V}}
\newcommand{\mU}{\mathcal{U}}
\newcommand{\me}{\mathcal{E}}
\newcommand{\mc}{\mathcal{C}}
\newcommand{\mg}{\mathcal{G}}
\newcommand{\mbe}{\mathbf{e}}
\def\eps{\epsilon}

\DeclareMathAlphabet{\mathpzc}{OT1}{pzc}{m}{it}

\newcommand {\ignore} [1] {}

\denseformat

\DeclareMathOperator{\defi}{\overset{\mathrm{def.}}{=} }

\date{}

\title{Near-Optimal Spanners for General Graphs in (Nearly) Linear Time}
\author{Hung Le}
\affil{University of Massachusetts Amherst}
\author{Shay Solomon}
\affil{Tel Aviv University}

\begin{document}
\pagenumbering{gobble}
\maketitle
\begin{abstract}
Let $G = (V,E,w)$ be a weighted undirected graph on $|V| = n$ vertices and $|E| = m$ edges, let $k \ge 1$ be any integer, and let
$\eps < 1$ be any parameter.
We present the following results on fast constructions of spanners with {near-optimal} sparsity and lightness,\footnote{The sparsity (respectively, lightness) is a normalized notion of size (resp., weight), where we divide the size (resp., weight) 
by the size $n-1$ of a spanning tree (resp., the weight $w(\mst)$ of a minimum spanning tree $\mst$).}
which culminate a long line of work in this area.
(By {\em near-optimal} we mean optimal under Erdos' girth conjecture and disregarding the $\eps$-dependencies.)
\begin{itemize}
\item  There are (deterministic) algorithms for constructing $(2k-1)(1+\eps)$-spanners for $G$ with a near-optimal sparsity of $O(n^{1/k} \cdot \log(1/\eps)/\eps))$.
The first algorithm can be implemented in the pointer-machine model within time $O(m\alpha(m,n) \cdot \log(1/\eps)/\eps) + \srt(m))$,
where $\alpha(\cdot,\cdot)$ is the two-parameter inverse-Ackermann function and $\srt(m)$ is the time needed to sort $m$ integers. The second algorithm can be implemented  in the \rdrm model  within time $O(m \log(1/\eps)/\eps))$. 
\item There is a (deterministic)  algorithm for constructing a $(2k-1)(1+\eps)$-spanner for $G$ that achieves a near-optimal bound of 
$O(n^{1/k} \cdot \poly(1/\eps))$
on both sparsity and lightness.
This algorithm can be implemented in the pointer-machine model within time $O(m\alpha(m,n) \cdot \poly(1/\eps) + \srt(m))$
and in the \rdrm model within time $O(m \alpha(m,n) \cdot \poly(1/\eps))$.
\end{itemize}
The previous fastest constructions of $(2k-1)(1+\eps)$-spanners with near-optimal sparsity incur a runtime of
is $O(\min\{m(n^{1+1/k}) + n\log n,k \cdot n^{2+1/k}\})$, even regardless of the lightness. 
Importantly, the {\em greedy spanner} for stretch $2k-1$ has sparsity $O(n^{1/k})$ --- with no $\eps$-dependence whatsoever,
but its runtime is $O(m(n^{1+1/k} + n\log n))$. Moreover, the   state-of-the-art lightness bound of any $(2k-1)$-spanner (including the greedy spanner) is poor, even regardless of the sparsity and runtime.
\end{abstract}

\pagebreak

\tableofcontents

\pagenumbering{arabic}

\clearpage
\section{Introduction}

Let $G = (V,E,w)$ be a weighted undirected graph on $|V| = n$ vertices and $|E| = m$ edges.
We say that $H$ is a $t$-spanner for $G$, for a parameter $t \ge 1$, if $H$ preserves all pairwise distances of $G$ to within a factor of $t$;
the parameter $t$ is called the \emph{stretch} of the spanner.
(A more detailed definition appears in~\Cref{sec:prelim}.) 
The most basic requirement from a low-stretch spanner is to be {\em sparse}, i.e., of small size;
the normalized notion of size, {\em sparsity}, is the ratio of the spanner size to the size $n-1$ of a spanning tree.
A generalized requirement is to have a small {\em weight}; the {\em weight} of a spanner is the sum of its edge weights, and the normalized notion of weight, {\em lightness}, 
is the ratio of the spanner weight to the weight $w(\mst(G))$ of a minimum spanning tree $\mst(G)$ for $G$.

Sparse and light spanners have been studied extensively over the years, and have found a wide variety of applications across different areas, from distributed computing and motion planning to  computational biology and machine learning.
As prime examples, they have been used in achieving efficient broadcast protocols~\cite{ABP90,ABP91}, for synchronizing networks and computing global functions \cite{Awerbuch85,PU89,Peleg00}, in gathering and disseminating data~\cite{BKRCV02,VWFME03,KV01}, and to routing~\cite{WCT02,PU89b,DBLP:conf/stoc/AwerbuchBLP89,TZ01}.

The holy grail is to achieve optimal tradeoffs between stretch and sparsity and between stretch and lightness, within a small running time.
For unweighted graphs, this goal has been achieved already in the mid 90s, via a simple yet clever clustering approach due to Halperin and Zwick~\cite{HZ96}: A linear-time construction of $(2k-1)$-spanners with the optimal (under Erdos' girth conjecture \cite{Erdos64}) {sparsity} of $O(n^{1/k})$;
we note that, for unweighted graphs, the sparsity and lightness parameters coincide.

The fundamental question underlying this work is whether one can achieve this goal in general weighted graphs. Chechik and Wulff-Nilsen \cite{CW16} gave a poly-time construction of $(2k-1)(1+\eps)$-spanners with 
a {\em near-optimal} bound of $O(n^{1/k} \cdot \poly(1/\eps))$ on both sparsity and lightness;
by {\em near-optimal} we mean optimal under Erd\H{o}s' girth conjecture and disregarding the $\eps$-dependencies. Although the runtime of the construction of \cite{CW16} is polynomial, it is far from linear.
Is it possible to achieve a fast --- ideally linear time --- spanner construction with the same guarantees?
This question is open even disregarding the lightness: All known spanner constructions with near-optimal sparsity incur a rather high runtime.

Next, we survey the main results on spanners for general graphs, starting with sparse spanners and proceeding to light spanners.
Subsequently, we present our contribution. 

\paragraph{Sparse spanners.~}
Graph spanners were introduced in the late 80s \cite{PS89,PU89}; initially, the focus was on the stretch-sparsity tradeoff.
For unweighted graphs, the aforementioned construction of \cite{HZ96} gives an optimal result.
We shall henceforth consider general $n$-vertex $m$-edge weighted graphs. 
The ``greedy spanner'' is perhaps the most basic spanner construction, introduced in the seminal work of
Alth\"{o}fer et al.~\cite{ADDJS93}.  
For any integer parameter $k \ge 1$, it provides a $(2k-1)$-spanner with sparsity $O(n^{1/k})$. On the negative side, the running time of the greedy spanner is rather high, namely $O(m(n^{1+1/k} + n \log n))$. 

The celebrated paper of Baswana and Sen \cite{DBLP:conf/icalp/BaswanaS03} presents a randomized algorithm for constructing $(2k-1)$-spanners with  sparsity $O(n^{1/k} \cdot k)$, within time $O(m \cdot k)$.
Roditty, Thorup and Zwick \cite{roditty2005deterministic} derandomized the Baswana-Sen \cite{DBLP:conf/icalp/BaswanaS03} algorithm, without any loss in parameters.
This result is optimal except for an extra factor of $k$ that appears in both the spanner size and the runtime bound.

Building on Miller et al.~\cite{MPVX15}, Elkin and Neiman \cite{EN18} gave a randomized algorithm for constructing $(2k-1)(1+\eps)$-spanners with sparsity $O(n^{1/k} \cdot \log k \cdot \log(1/\eps)/\eps)$, within time $O(m)$, for any $\eps < 1$; 
in fact, their runtime analysis overlooks the time consumed by a certain bucketing procedure, which, at least naively, requires $\Omega(\srt(m))$ time, where $\srt(m)$ is the time needed to sort $m$ integers. 
Alstrup et al. \cite{ADFSW19} achieved a deterministic algorithm with the same guarantees; we note that time $\Omega(\srt(m))$ is also needed by the construction of \cite{ADFSW19} for the same reason. 
These results demonstrate that by incurring an arbitrarily small multiplicative error of $1+\eps$ to the stretch bound, one can achieve, within linear time (modulo the overlooked time needed for integer sorting), a near-optimal sparsity bound, except for an extra $\log k$ factor. Additional results are summarized in \Cref{tab:spanners}.
 
\begin{table}[h]
	\small
	\centering
	\mbox{
		\begin{tabular}{c|c|c|c|c}
			\toprule
			\bf Stretch & \bf Sparsity & \bf Lightness & \bf Construction Time & \bf Ref\\
			\midrule
			$(2k-1)$&   $O \left(n^{1/k} \right)$ & $O\left(n/k\right)$ & $O \left(mn^{1+1/k} + n\log n\right)$ & \cite{ADDJS93}\\
			$(2k-1)(1+\eps)$&   $O \left(n^{1/k} \right)$ & $O \left( k
			n^{1/k}\right)$ & $O \left(mn^{1+1/k}  + n\log n \right)$ & \cite{CDNS92}\\
			$(2k-1)$ & $O(n^{1/k})$ & --- & $O \left( kn^{2+1/k} \right)$ & \cite{RZ11}\\	

			$(2k-1)$&   $O \left(k \cdot n^{1/k} \right)$ & --- & $O \left(kmn^{1/k}\right)$ & \cite{TZ01b}{$^{R}$}\\			
			$(2k-1)(1+\eps)$&   $O \left(n^{1/k} \right)$ & $O \left(kn^{1/k}\right)$ & $O \left( kn^{2+1/k} \right)$ & \cite{ES16}
			\\								
			$(2k-1)(1+\eps)$&   $O \left(n^{1/k} \right)$ & $O \left(n^{1/k}
			\cdot k/\log k \right)$ & $O \left(mn^{1+1/k}  + n\log n \right)$ & \cite{DBLP:journals/siamdm/ElkinNS15}$$\\
			$(2k-1)(1+\eps)$&   $O \left(n^{1/k} \right)$ & $O\left(n^{1/k}\right)$  & $n^{\Theta(1)}$ & \cite{CW18}\\
			$(2k-1)(1+\eps)$&   $O \left(n^{1/k} \right)$ & $O \left(n^{1/k}\right)$ & $O \left(mn^{1+1/k}  + n\log n\right)$ & \cite{FS20}\\
            $(2k-1)(1+\eps)$ & $O\left(n^{1/k}\right)$ &
            $O\left(n^{1/k}\right)$ &
            $O(n^{2+1/k+\eps'})$ & \cite{ADFSW19}\\
			\midrule
			$(2k-1)$&   $O \left(k \cdot n^{1/k} \right)$ & --- & $O \left(km\right)$ & \cite{baswana2007simple}{$^{R}$} \cite{RTZ05}\\     
			$(2k-1)(1+\eps)$&   $O \left(k \cdot n^{1/k} \right)$ & $O \left(kn^{1/k}\right)$ & $O \left(\srt(m)+km + n \log n \right)$ & \cite{ES16}\\	
			$(2k-1)(1+\eps)$ &  $O(\log k \cdot n^{1/k})$ &  $O \left(k \cdot n^{1+1/k} \right)$ & $O(\srt(m)+n \cdot \log n)$ & \cite{EN18}{$^{R}$}\\				  			 
			$(2k-1)(1+\eps)$ & $O \left(\log k\cdot n^{1/k} \right)$ & --- & $O(\srt(m))$ & \cite{EN18,ADFSW19}\\			
			$(2k-1)(1+\eps)$ & $O \left(\log k \cdot n^{1/k} \right)$ & $O \left( \log k \cdot n^{1/k}  \right)$ & $O(\srt(m) + n \cdot\log n)$ & \cite{ADFSW19}\\

			\midrule
			\midrule
$(2k-1)(1+\eps)$& $O(n^{1/k})$& --- & $O(m\alpha(m, n)+ \srt(m))$ & \Cref{thm:1}$\text{ }^{P}$\\
$(2k-1)(1+\eps)$ & $O(n^{1/k})$ & --- &  $O(m)$ & \Cref{thm:1}$\text{ }^{W}$\\
$(2k-1)(1+\eps)$& $O(n^{1/k})$ &$O(n^{1/k})$ & $O(m\alpha(m, n) + \srt(m))$& \Cref{thm:2}$\text{ }^{P}$\\
$(2k-1)(1+\eps)$& $O(n^{1/k})$ & $O(n^{1/k})$& $O(m\alpha(m, n))$& \Cref{thm:2}$\text{ }^{W}$\\
\bottomrule
		\end{tabular}
	}
\caption{Table summarizing known and new spanner constructions for general weighted graphs, for stretch values of $2k-1$ and $(2k-1)(1+\eps)$. In the top and middle parts of the table we list rather slow and fast known spanner constructions, respectively.
Our new results appear at the bottom.  Results marked with { $^{R}$} correspond to randomized constructions. We use the superscript marks $\text{}^{P}$  and $\text{}^{W}$ to distinguish between the new algorithms 
that apply to the pointer-machine   versus the \rdrm models, respectively.}
\label{tab:spanners}
\end{table}

As shown in \Cref{tab:spanners}, the previous state-of-the-art runtime for constructing $(2k-1)(1+\eps)$-spanners with near-optimal sparsity
is $O(\min\{m(n^{1+1/k}) + n \log n,k \cdot n^{2+1/k}\})$, even regardless of the lightness.

\begin{question} \label{q1}
Can one achieve a (nearly) linear time spanner construction with near-optimal sparsity?  
\end{question}

We answer~\Cref{q1} in the affirmative by presenting two algorithms for constructing spanners with near-optimal sparsity in near-linear time. Specifically, we prove the following result. (Refer to
 \Cref{tab:spanners} for a detailed comparison between our and previous results.)
\begin{theorem}\label{thm:1}
For any weighted undirected $n$-vertex $m$-edge graph $G$, any integer $k \ge 1$ and any $\eps < 1$, one can deterministically construct  $(2k-1)(1+\eps)$-spanners with a near-optimal sparsity of $O(n^{1/k} \cdot \log(1/\eps)/\eps)$. This construction can be implemented: 
\begin{itemize}
\item In the pointer-machine model within time $O(m\alpha(m,n) \cdot \log(1/\eps)/\eps + \srt(m))$.\footnote{In the \emph{pointer machine model}, one can perform binary comparisons between data, arithmetic operations on data, dereferencing of pointers, and equality tests on pointers. The model does not permit pointer arithmetic or tests other than equality on pointers and thus is less powerful than the RAM model \cite{DBLP:journals/jcss/Tarjan79}.}

\item In the \rdrm model 
within time $O(m \log(1/\eps)/\eps)$.\footnote{The \rdrm model is similar to the classic unit-cost RAM model, except that (1) For a word length $w \ge 1$ the contents of all memory cells are integers up to $2^w$. (2) Some additional instructions are available; in particular, the available unit-time operations are those from the \emph{restricted instruction set}: addition and subtraction, (noncyclic) bit shifts by an arbitrary number of positions, and bitwise boolean operations, but not multiplication. (3) It is also assumed that $w \ge \log{n}$.
We note that if the running time of the algorithm depends on the input size but not on the word size,
then the model is further called \wrdram model; the running time of our algorithms do not depend on the word size.} 
\end{itemize}
\end{theorem}

We remark that $\alpha(m,n) = O(1)$ when $m = \Omega(n\log^{*}n)$. In fact, $\alpha(m,n) = O(1)$ even when $m = \Omega(n\log^{*(c)}n)$ for any constant $c$, where $\log^{*(\ell)}(.)$ denotes the iterated log-star function with $\ell$ stars; that is, $O(m\alpha(m,n))$ is bounded by $O(m  + n\log^{*(c)}{n})$ for any constant $c$.
Thus the running time in the first item of~\Cref{thm:1} is linear in $m$ in almost the entire regime of graph densities, i.e., except for very sparse graphs.
Moreover, even when $\alpha(m,n)$ is super-linear, it can still be viewed as constant for most practical purposes. 
However, there is a significant qualitative difference between truly linear-time and nearly linear-time algorithms,
and shaving this factor for the entire regime of graph densities, as provided by the second item of~\Cref{thm:1}, is of fundamental theoretical importance.

The previous linear-time algorithms for constructing sparse spanners in general weighted graphs 
\cite{MPVX15,EN18,ADFSW19} achieve a sub-optimal sparsity bound of $O(n^{1/k} \cdot \log k \cdot \log(1/\eps)/\eps)$, and, as mentioned, their runtime is actually $O(\srt(m))$.
Moreover, these constructions, as well as all other spanner constructions with runtime $o(k m)$ (including ours), use a hierarchical clustering approach that involves constructing a so-called {\em cluster graph} in each level of the hierarchy. Importantly, the cluster graph is a simple graph (without self loops and parallel edges), and all the previous works either overlooked the time needed to guarantee that the cluster graph is simple or they included an extra factor of $\alpha(m,n)$ in the runtime bound --- due to the usage of the classic 
\djset~data structure~\cite{Tarjan75}.  
We demonstrate that this factor can be shaved via a novel clustering approach, 
which we name {\em MST-clustering}; refer to \Cref{tech} for a discussion on the technical details.

\paragraph{Light spanners.~}
Like sparsity, the lightness of spanners has been extremely well-studied. Alth\"{o}fer et al.~\cite{ADDJS93} showed that the lightness of the greedy $(2k-1)$-spanner is $O(n/k)$. 
Despite extensive research, the state-of-the-art lightness bound of any known $(2k-1)$-spanner construction (including the greedy spanner) remains poor, even regardless of the sparsity and runtime. It is thus only natural to explore the lightness bound for a slightly increased stretch of $(2k-1)(1+\eps)$, where $\eps < 1$ is an arbitrarily small parameter of our choice.
Chandra et al.~\cite{CDNS92} showed that the greedy $(2k-1)(1+\eps)$-spanner has lightness $O(k \cdot n^{1/{k}} \cdot (1/\eps)^{2})$.
There was a sequence of works from recent years on light spanners \cite{ES16,ENS14,CW16,FS20,EN18,ADFSW19,LS21}. 
In particular, a construction of $(2k-1)(1+\eps)$-spanners with a near-optimal lightness of $O(n^{1/k} \cdot \poly(1/\eps))$ within a runtime of $O(m \alpha(m,n)$ was presented recently \cite{LS21}, where $\alpha(\cdot,\cdot)$ is the inverse-Ackermann function;
on the negative side, the sparsity of the construction of \cite{LS21} is unbounded.
As mentioned, the construction of \cite{CW16} achieves a {near-optimal} bound of $O(n^{1/k} \cdot \poly(1/\eps))$ on both sparsity and lightness, but its runtime is far from linear. The result of Filtser and Solomon~\cite{FS20} implies that the greedy spanner achieves the same bounds
as the construction of \cite{CW16}, but the runtime $O(m(n^{1+1/k} + n \log n))$ of the greedy spanner  is also rather high. 

\begin{question} \label{q2}
Can one achieve a (nearly) linear time spanner construction with a near-optimal bound on both the sparsity and lightness?
\end{question}

We answer~\Cref{q2} in the affirmative by presenting an algorithm for constructing $(2k-1)(1+\eps)$-spanners with near-optimal sparsity and lightness in near-linear time, which culminates a long line of work in this area. Specifically, we prove the following result.

\begin{theorem}\label{thm:2}
For any weighted undirected $n$-vertex $m$-edge graph $G$, any integer $k \ge 1$ and any $\eps < 1$, one can deterministically construct  $(2k-1)(1+\eps)$-spanners with a near-optimal bound of $O(n^{1/k} \cdot \poly(1/\eps))$ on both sparsity and lightness. This construction can be implemented: 
\begin{itemize}
\item In the pointer-machine model within time $O(m\alpha(m,n) \cdot \poly(1/\eps) + \srt(m))$.
\item In the \rdrm model within time $O(m \alpha(m,n) \cdot \poly(1/\eps))$.
\end{itemize}
\end{theorem}

We obtain the result of~\Cref{thm:2} by strengthening the framework of~\cite{LS21} for fast constructions of light spanners
to achieve a near-optimal bound on the sparsity as well.
To this end, we plug the ideas used in the proof of~\Cref{thm:1}, in conjunction with numerous new insights, on top of the framework of~\cite{LS21} in a highly nontrivial way. Our MST-clustering approach plays a key role not just in the proof of~\Cref{thm:1},
but also in the proof of~\Cref{thm:2};
refer to \Cref{tech} for more details.  

\subsection{Technical Highlights} \label{tech}

Our spanner construction is inspired by the constructions of~\cite{MPVX15},~\cite{EN18} and~\cite{ADFSW19}, which we briefly review next. 
All these constructions achieve a runtime of $O(m)$, modulo the time needed for sorting the edge weights; we shall elaborate on this point later. 
The construction of~\cite{MPVX15} achieves stretch $O(k)$ with sparsity $O(n^{1/k}\log(k)))$, while the two other constructions achieve the same sparsity but with a stretch of $(2k-1)(1+\eps)$. (For clarity, we shall ignore the dependency on $\eps$ in the sparsity bounds.) 

The construction of~\cite{MPVX15}\footnote{The algorithm used in~\cite{MPVX15} is parallel, and our interpretation of it is in the standard sequential model.} starts by dividing the edge set into $\mu_k \defi O(\log k)$ sets $\{E^1,E^2,\ldots, E^{\mu_k}\}$, such that for each set $E^{\sigma}$, $\sigma \in [1,\mu_k]$, any two edge weights are either within a factor of $2$ from each other or they are separated by at least a factor of $k^{c}$ for some constant $c$. The algorithm then focuses on  constructing a spanner $H^{\sigma}$ for each edge set $E^{\sigma}$ separately; the sparsity of $H^{\sigma}$  is  $O(n^{1/k})$,  which ultimately leads to a sparsity bound of $O(\mu_k \cdot n^{1/k}) = O(\log(k)\cdot n^{1/k})$  of the final spanner $H$. In the construction of $H^{\sigma}$, the edge set $E^{\sigma}$ is further divided into smaller subsets $\{E^{\sigma}_1,E^{\sigma}_2,\ldots\}$, where edges in the same set $E^{\sigma}_i$ have the same weights up to a factor of 2, and the weights of edges in $E^{\sigma}_i$  are at least $k^{c}$ times greater than the weights of edges in  $E^{\sigma}_{i-1}$, for each $i$.  
The construction of~\cite{MPVX15} uses a \emph{hierarchy of clusters} and an \emph{unweighted} {\em cluster graph} $R_i$  for each level $i$ of the hierarchy. The vertex set of $R_i$ corresponds to a \emph{subset} of level-$i$ clusters that are incident to at least one edge in $E_i^{\sigma}$, and the edge set of $R_i$ corresponds to a subset of edges in $E^{\sigma}_{i}$  
 interconnecting level-$i$ clusters. A preprocessing step is applied to the construction of $R_i$ to remove parallel edges, which are edges in $E^{\sigma}_{i}$ connecting the same two level-$i$ clusters, and self-loops, which are edges in $E^{\sigma}_{i}$  whose both endpoints are in the same  level-$i$ cluster. The construction of~\cite{MPVX15} then
builds an $O(k)$-spanner for the (unweighted) graph $R_i$ to obtain a subset of edges $S_i$ of $E^{\sigma}_{i}$ to add to $H^{\sigma}$. Next, vertices in $R_i$ are grouped into a set $\mathcal{U}$ of subgraphs of (unweighted) diameter $\Theta(k)$; each subgraph in $\mathcal{U}$ is then transformed into a level-$(i+1)$ cluster. The construction then continues to level $i+1$, then to level $i+2$, etc., until all the edges in the graph have been considered. The construction of the (unweighted) $O(k)$-spanner of $R_i$ and the set of subgraphs $\mathcal{U}$ is randomized and based on sampling from an exponential distribution.   

The construction of~\cite{EN18} builds on that of~\cite{MPVX15}. First, it partitions the edge set into $\mu_{k,\eps} = O(\log(k)/\eps)$ sets of edges instead of $O(\log k)$ sets as in~\cite{MPVX15}; the idea is that for each set $E^{\sigma}$, $\sigma \in [1,\mu_{k,\eps}]$, any two edge weights are either within a factor of $1+\eps$ from each other, or are separated by at least a factor of $k^{c}$ for some constant $c$. Next, the construction of~\cite{EN18} uses the same idea of~\cite{MPVX15} to construct the spanner of $R_i$  and the set of subgraphs $\mathcal{U}$. However, the stretch of the spanner is improved to $(2k-1)$, which readily implies a stretch of $(4k-2)(1+\eps)$  for the final spanner. We note that the stretch is  $(4k-2)(1+\eps)$  instead of $(2k-1)(1+\eps)$,  due to a subtlety involving randomness in~\cite{MPVX15}. With a more sophisticated analysis, \cite{EN18} resolves this subtlety and reduces the stretch to $(2k-1)(1+\eps)$. The sparsity of the final spanner is $O(\mu_{k,\eps} \cdot n^k) = O(\log(k) \cdot n^{1/k})$, ignoring the dependence on $\eps$. 

Unlike the constructions of \cite{MPVX15,EN18}, the construction of~\cite{ADFSW19} is \emph{deterministic}. A central idea in the construction of~\cite{ADFSW19}, inspired by an earlier work~\cite{ES16}, is to use a modified version of the Halperin-Zwick algorithm \cite{HZ96} in the construction of the spanner of $R_i$. The spanner of $R_i$ has stretch $(2k-1)$, which implies the final stretch of $(2k-1)(1+\eps)$. The sparsity of the spanner remains $O(\mu_{k,\eps} \cdot n^k) = O(\log(k) \cdot n^{1/k})$, as in \cite{MPVX15,EN18}.

We note the following points regarding the aforementioned constructions. 
\begin{enumerate}
\item First, the sparsity incurs an extra factor of $O(\log k)$,
i.e., it is $O(\log(k) \cdot n^{1/k})$ rather than $O(n^{1/k})$. This is inevitable, since subgraphs in $\mathcal{U}$ of $R_i$ have a diameter of $\Theta(k)$, hence the weights of  edges in $E^{\sigma}_{i+1}$ and $E^{\sigma}_{i}$ must be at least a factor of $k^c$ apart from each other,  which ultimately leads to a factor $O(\log k)$ in the number of sets that the edge set $E$ is partitioned to. 
\item Second, each set $E^{\sigma}$ is partitioned into  $O(\log U)$ sets $\{E_1^{\sigma}, E_2^{\sigma},\ldots\}$,  where $U$ is the maximum edge weight. Thus, at least naively, the partition of  $E^{\sigma}$ can be constructed in time $O(m + \log U)$ rather than $O(m)$, where $U$ could be  unbounded. One way to avoid the dependency on $U$ is to sort all edge weights of $E^{\sigma}$, which requires time $\srt(m)$. We note that the computation of the partition of $E^{\sigma}$ into subsets has been overlooked in 
the aforementioned constructions ~\cite{MPVX15,EN18,ADFSW19}. In the \rdrm model, we use the simple observation that $O(\log U)$ is roughly the word size to guarantee that such a partition can be computed within $O(m)$ time.  
\item Third, the aforementioned constructions involve constructing a cluster graph $R_i$ associated with each level $i$ of the hierarchy. 
While the details of maintaining $R_i$ are not precisely described in these constructions, we observe that $R_i$ can be  efficiently maintained using  the \djset~data structure. However, the total runtime would be $O(m\alpha(m,n))$ rather than $O(m)$. 
We next show that the non-optimal sparsity bound of $O(n^{1/k}\log k)$ achieved by the previous works can be used to remove the factor $\alpha(m,n)$. Observe that $m \alpha(m,n) = O(m)$ when $m = \Omega(n \log\log(n))$. If $m = O(n^{1+1/k}\log k)$, we can simply return the whole graph as the output spanner. Otherwise, $m = \Omega(n^{1+1/k}\log k) = \Omega(n \log\log n)$ for every $k\geq 2$, in which case the total   time to construct a spanner of size $O(n^{1+1/k}\log k)$ is $O(m \alpha(m,n)) = O(m)$. However, the same argument fails when aiming for the near-optimal sparsity bound of $O(n^{1/k})$ that we achieve (e.g., $O(n^{1+1/k}) = O(n)$ when $k = \Omega(\log n)$). To construct a spanner with a sparsity of $O(n^{1/k})$ in $O(m)$ time, one must overcome the ``\djset~barrier''. We note that even in the cell-probe model, which is stronger than the \rdrm model, one cannot avoid the factor $\alpha(m,n)$ in the \djset~data structure~\cite{FS89}.  
\end{enumerate}

 Our first construction is in the pointer-machine model; there we overcome the ``(unweighted) diameter barrier'' of $\Theta(k)$ of subgraphs in $\mathcal{U}$ constructed from $R_i$: Subgraphs in our construction have (unweighted) diameters of $O(1)$. As a consequence, we demonstrate that it suffices to partition $E$ into $\mu_{\eps} \defi O(\frac{1}{\eps}\log(1/\eps))$ sets instead of $O(\log k/\eps)$ sets, which ultimately leads to the optimal sparsity of $O(n^{1/k})$, ignoring the dependence on $\eps$. The key idea behind our construction is rather simple --- we prove that it suffices to construct level-$(i+1)$ clusters from level-$i$ clusters such that the total number of clusters is reduced by $\Omega(|V(R_i)|)$. We then use the Halperin-Zwick algorithm~\cite{HZ96} to construct a $(2k-1)$-spanner for $R_i$. Next, we construct the set of subgraphs $\mathcal{U}$ greedily, with each having diameter $O(1)$. By using the \djset~data structure in the construction of $R_i$, the total running time of our algorithm is $O(m\alpha(m,n))$, plus an additive term of $\srt(m)$ needed for computing the partition of $E^{\sigma}$ as discussed above. Note that we cannot use the trick that we provided earlier to remove the $\alpha(m,n)$ factor since our spanner construction does not have any slack on the sparsity.
Our construction is deterministic, it improves the aforementioned constructions~\cite{MPVX15,EN18,ADFSW19} ---
yet is arguably simpler.  

Our linear-time spanner construction in the \rdrm model is based on a novel clustering approach, which we name {\em MST-clustering}. Specifically, we guarantee that the subgraphs induced by clusters are subtrees of a minimum spanning tree (MST) of the graph, denoted by $\mst$, and hence, every \union~operation is performed along the edges of $\mst$. That is, each \union~operation is  of the form $\union(u,v)$, where $(u,v)$ is an edge in $\mst$. As a result, we are able to determine all the \union ~operations even before the cluster construction takes place. This allows us to use a refined \djset~data structure, by Gabow-Tarjan \cite{GT85}, which has $O(1)$ amortized cost per \union/\find~operation.  To the best of our knowledge, this is the first time that the MST serves as the union tree in the Gabow-Tarjan \djset~data structure, other than in applications that {\em directly concern MST}.

The idea of using the MST in the context of clustering in spanner constructions is quite surprising. In many of the known spanner constructions, clusters in the cluster hierarchy need to satisfy a diameter constraint. That is, clusters at level-$i$ should have a diameter of at most $f(L_i)$, for some function $f$, often a linear function, where $L_i$ is an upper bound on the edge weights in $E^\sigma_i$. 
In particular, the approaches of~\cite{MPVX15,EN18,ADFSW19,LS21} utilize the fact that some edges (not in $\mst$) have been added during the construction of clusters at lower levels, and use these edges to construct clusters that satisfy the diameter constraint.  By restricting ourselves to only use $\mst$ for clustering, it seems much more challenging (and perhaps impossible at first) to guarantee the diameter constraint for level-$i$ clusters. Our key insight is that it is still possible to do so, and to this end we rely on the cycle property of MST, both for arguing that clusters have small diameters and for constructing clusters efficiently. 
 
Finally, we show how to construct a spanner with near-optimal sparsity {\em and lightness}. Our construction builds on the fast construction of spanners with near-optimal lightness in~\cite{LS21}. The construction of \cite{LS21} has a preprocessing step and a main construction step. In the preprocessing step, every edge of weight at most $\frac{w(\mst)}{m\eps}$ is added to the spanner. Clearly the number of edges added in this step could be as large as $\Omega(n^2)$ (for dense graphs). Our first observation is that, except for $\mst$ edges, edges added in the preprocessing step  are not involved in the main construction step, and hence we can apply our sparse spanner construction from \Cref{thm:1} to reduce the number of edges added in the preprocessing step to $O(n^{1+1/k})$. The main construction step is based on a cluster hierarchy. However, clusters in~\cite{LS21} are ``equipped'' with a potential function, and the challenge of the cluster construction is to guarantee a sufficient reduction in  the potential values between two consecutive levels of the hierarchy. A cluster graph is also used to select a subset of edges in $E^\sigma_i$ to add to the spanner. Again, the number of edges added in this step could be as large as $\Omega(n^2)$. In order to obtain a spanner with near-optimal guarantees on both sparsity and lightness, we employ the insight that we developed in this paper for the construction of sparse spanners, by constructing clusters in such a way that, between two consecutive levels, there is a sufficient reduction not just in the potential values, but also in the {\em number of clusters}. This, in turn, makes the task of constructing clusters much more challenging; indeed, a-priori, it is unclear that it is possible to achieve both objectives via a single (fast) spanner construction.

The spanner construction of \cite{LS21} constructs level-$(i+1)$ clusters in 5 steps; each level-$(i+1)$ cluster corresponds to a subgraph of a cluster graph $R_i$. We note that the cluster graph $R_i$ in this construction is different from the cluster graph used in the sparse spanner constructions in that its MST, denoted by $\msttilde_{i}$, is derived from the MST of $G$. We observe that among the 5 steps used in \cite{LS21}, there are two steps where the reduction in the number of clusters is not guaranteed. Furthermore, the  clusters  formed in these two steps are subgraphs of $\msttilde_{i}$. Thus, our idea is to apply the insights we developed in the sparse spanner construction in the \rdrm model to this setting. However, there are two subtleties in the construction of \cite{LS21} that we need to address. First, the cluster graph $R_i$ has weights on both edges and vertices. As a result, $\msttilde_{i}$ also has weights on both edges on vertices. Second, clusters in the construction of \cite{LS21} contain \emph{virtual vertices}; these vertices are not in the input graph and are introduced to support the design of the potential function for clusters. We show an analogous version of the cycle property for $\msttilde_{i}$. We use this property, in addition to several other technical ideas, to transfer insights that we developed in the construction of sparse spanners in the \rdrm model to the cluster construction in this setting. As a result, our spanner construction that achieves near-optimal bounds on both sparsity and lightness is much more involved than our two aforementioned constructions 
(which prove \Cref{thm:1}) with near-optimal sparsity but possibly huge lightness.

\section{Preliminaries}\label{sec:prelim}

We denote by $G = (V,E,w)$ a graph $G$ with vertex set $V$, edge set $E$, and weight function $w: E(G)\rightarrow \mathbb{R}^+$ on its edges. We denote by $\mst(G)$ the minimum spanning tree of $G$; there could be MSTs for $G$, but we may assume w.l.o.g.\ that there is only one (e.g., by using lexicographic rules to break ties for edges of the same weight). When the graph is clear from the context, we abbreviate $\mst(G)$ as $\mst$. We denote by $w(G) = \sum_{e\in E}w(e)$ the {\em weight} of $G$, i.e., the sum of all edge weights in $G$. 

We use $d_G(u,v)$ to denote the distance between two vertices $u$ and $v$ in $G$. 
The diameter of $G$ is the maximum pairwise distance in $G$, and is denoted by $\dm(G)$.

For a subset of vertices $X\subseteq V$, we denote by $G[X]$ the subgraph of $G$ induced by $X$. We also define a subgraph of $G$ induced by an \emph{edge set} $F$ by $G[F] = (V, F)$ 

Let $H$ be a spanning subgraph of $G$ (with edge weights inherited from $G$). The {\em stretch} of $H$ is defined as $\max_{u\not\in v \in V}\frac{d_H(u,v)}{d_G(u,v)}$; $H$ is called a {\em $t$-spanner} of $G$ if its stretch is at most $t$. 
The next well-known observation,
which states that the stretch of $H$ is realized by an edge of $G$, 
  follows from the triangle inequality.
\begin{observation} \label{stretch:ob}
	$\max_{u \not= v \in V(G)} \frac{d_H(u,v)}{d_G(u,v)} =  \max_{(u,v) \in E(G)} \frac{d_H(u,v)}{d_G(u,v)}$.
\end{observation}
 We say that $H$ is a spanner for a \emph{subset of edges $X\subseteq E$} if $\max_{(u,v) \in X} \frac{d_H(u,v)}{d_G(u,v)} \leq t$. 
 
Our constructions use the aforementioned linear-time construction of $(2k-1)$-spanners for unweighted graphs by Halperin-Zwick~\cite{HZ96}, which we record in the following theorem for further use.

\begin{theorem}[\cite{HZ96}]\label{thm:HalperinZwick} For any unweighted   $n$-vertex $m$-edge graph $G$ and any integer $k \ge 1$, a $(2k-1)$-spanner of $G$ with $O(n^{1+\frac{1}{k}})$ edges can be constructed deterministically in $O(m + n)$  time.
\end{theorem}

\section{An $O(m\alpha(m,n) + \srt(m))$-time Algorithm}\label{sec:PointerMachine}

In this section we prove the first item of~\Cref{thm:1}. By scaling, we assume that the minimum edge weight is $1$. We partition the edge set $E$ into $\mu_\epsilon = \log_{1+\eps} \left(\frac{1}{\eps}\right) = \Theta(\frac{\log(1/\eps)}{\eps})$  sets $\{E^{\sigma}\}_{1\leq \sigma \leq \mu_\epsilon}$ such that each $E^{\sigma}$ can be written as $E^{\sigma} = \cup_{i\in \mathbb{N}^+}E^{\sigma}_i$ with: 
\begin{equation} \label{eq:EsigmaI}
	E^{\sigma}_i = \{e \in E: \frac{L_i}{(1+\eps)} \leq  w(e) \leq L_i, i \in \mathbb{N}\} \mbox{, where } L_i = L_0/\eps^{i}, L_0 = (1+\eps)^{\sigma}.
\end{equation} 
Thus, for any edge set $E^{\sigma}$, any two edge weights are either roughly the same (up to a factor of $1+\eps$) or 
far apart (separated by at least a factor of $1/\eps$).
For technical convenience, we shall define $L_{-1} = 0$.

We note that the time needed to compute the partition of $E$ into the sets $\{E^{\sigma}\}_{1\leq \sigma \leq \mu_\epsilon}$
is upper bounded by $O(m + \srt(m')) = O(\srt(m))$, where $m'$ is the number of non-empty sets.
Indeed, this computation can be carried out naively in linear time, except for the time needed to sort the {\em indices} of the non-empty sets in $\{E^{\sigma}_i\}_{1\leq \sigma \leq \mu_\epsilon, i \in \mathbb{N}}$. In the runtime analysis that follows we shall disregard this initial time investment, under the understanding that we include it in the final runtime bound.
 
We now construct a $(2k-1)(1+O(\eps))$-spanner $H^{\sigma}$ for each set $E^{\sigma}$ with sparsity $O(n^{1/k})$ in $O(m \cdot \alpha(m,n))$ time. A  $(2k-1)(1+O(\eps))$-spanner $H$ for $G$ with sparsity $O(n^{1/k} \cdot \frac{\log(1/\eps)}{\eps})$ is then obtained as the union of all $H^{\sigma}$'s: $H = \cup_{1\leq \sigma \leq \mu_\epsilon}H^{\sigma}$, within time $O(m \cdot \alpha(m,n) \cdot \frac{\log(1/\eps)}{\eps})$.

We focus on the construction of $H^{\sigma}$, for a fixed $\sigma \in [1,\mu_{\eps}]$. Initially $H^{\sigma}_{0} = (V,\emptyset)$. The construction is carried out in what we call \emph{levels}: at level $i$, we shall construct a subgraph $H^{\sigma}_i$ such that $H^{\sigma}_{\leq i}$ is a $(2k-1)(1+O(\eps))$-spanner for the edge set $E^{\sigma}_{\leq i}$. Here $H^{\sigma}_{\leq i} = \cup_{0\leq j \leq i}H^{\sigma}_j$ and  $E^{\sigma}_{\leq i} =  E^{\sigma}_{0\leq j \leq i}$.  By induction, $H^{\sigma} \defi \cup_{i\geq 0}H^{\sigma}_i$ would provide a $(2k-1)(1+O(\eps))$-spanner for the edge set $E^{\sigma}$. 
Consequently,  $H = \cup_{1\leq \sigma \leq \mu_\epsilon}H^{\sigma}$ will provide a
$(2k-1)(1+O(\eps))$-spanner for $E = \bigcup_{1\leq \sigma \leq \mu_\epsilon} E^{\sigma}$,
and, by~\Cref{stretch:ob}, also for $G$ .
All graphs $H^{\sigma}_i$  share the same vertex set $V$ and hence are distinguished by the edge set.

A {\em cluster} is a set of vertices.
Our construction uses a {\em hierarchical clustering}, where for each $i \ge 0$, the construction at level $i$ is associated with a set of { clusters}
$\mathcal{C}_i$ such that: 
\begin{itemize}[noitemsep]
\item \textbf{(P1)~} 	\hypertarget{P1}{}
	Each cluster $C\in \mathcal{C}_i$ is a subset of $V$. Furthermore, clusters in $\mathcal{C}_i$ induce a partition of $V$.
\item \textbf{(P2)~} 	\hypertarget{P2}{}
Each cluster $C\in \mathcal{C}_i$ induces a subgraph $H_{\leq i}^{\sigma}[C]$ of diameter $\le gL_{i-1}$ for some constant $g$. 
\end{itemize}

$\mathcal{C}_0$ is the set of $n$ singletons of $V$ and hence trivially satisfies both Properties (\hyperlink{P1}{P1}) and (\hyperlink{P2}{P2}) (recall that $L_{-1} = 0$). The cluster sets $\{\mathcal{C}_0, \mathcal{C}_1, \ldots\}$ provide a {\em hierarchy of clusters} $\mathcal{H}$. In particular, for any $i \ge 1$, $\mathcal{C}_{i-1}$ is a {\em refinement} of $\mathcal{C}_i$: any cluster $C\in \mathcal{C}_i$ is the union of a subset of clusters in $\mathcal{C}_{i-1}$.

\paragraph{Representing $\mathcal{C}_i$ by Disjoint Sets.~} We shall use the classic \djset~data structure~\cite{Tarjan75} in our clustering procedure, for representing clusters in $\mathcal{C}_i$,  grouping subsets of clusters to larger clusters (via the \textsc{Union} operation), and checking whether a pair of vertices belongs to the same cluster (via the \textsc{Find} operation). In particular, each cluster $C\in \mathcal{C}_i$ will have a {\em representative} vertex, denoted by $r(C)$, that can be accessed from any vertex $v\in C$ by calling $\find(v)$; we define $r(v) := \find(v)$.  
The {\em amortized} time per each \textsc{Union} or \textsc{Find} operation is $O(\alpha(a,b))$, where $a$ is the total number of \textsc{Union} and \textsc{Find} operations  and $b$ is the number of vertices in the data structure.

\paragraph{Constructing $H^{\sigma}_i$.~} We assume that $|E^{\sigma}_i| \geq 0$; otherwise, we will skip the construction at level $i$ and set $\mathcal{C}_{i+1} = \mathcal{C}_i$. We say that a cluster at level $i$ is {\em isolated} if none of its vertices is incident on any edge of $E^{\sigma}_i$;
otherwise it is \emph{non-isolated}. Let $\mx$ be the set of all non-isolated level-$i$ clusters.  We say that two edges $(u,v)$ and $(u',v')$ in $E_i^{\sigma}$ are {\em parallel} if $r(u) = r(u')$  
(i.e., $u$ and $u'$ are in the same level-$i$ cluster) and $r(v) = r(v')$ (i.e., $v$ and $v'$ are in the same level-$i$ cluster). We say that $(u,v)$ is a self-loop if $r(u) = r(v)$ (i.e., $u$ and $v$ are in the same level-$i$ cluster). Let $S_i$ be obtained from $E^{\sigma}_i$ by removing from it all self-loops and keeping only the lightest edge in every maximal set of parallel edges of $E^{\sigma}_i$;
we refer to the edges of $S_i$ as the {\em source edges}. 

We then construct an unweighted graph $R_i$, called the \emph{representative graph}, as follows: $V(R_i) = \{r(C): C\in \mx\}$ and 
$E(R_i) =  \{(r(u),r(v)): (u,v) \in S_i\}$. The vertices and edges of $R_i$ are referred to as the {\em representative vertices} and {\em representative edges}, respectively;
note that each representative edge corresponds to a unique source edge.  Let $E'_i \leftarrow \textsc{HalperinZwick}(R_i,2k-1)$ be the edge set obtained by applying the spanner algorithm of~\Cref{thm:HalperinZwick} to $R_i$. Let $S_i'$ be the subset of source edges in $S_i$ corresponding to  the representative edges in $E_i'$.  Our graph  $H^{\sigma}_i$ has $S_i'$ as its edge set.

\begin{lemma}\label{lm:Stretch} $d_{H^{\sigma}_{\leq i}}(u,v) \leq (2k-1)(1+O(\eps))w(u,v)$ for every edge $(u,v) \in E^{\sigma}_i$, assuming $\eps \leq 1/(2g)$. Furthermore, $S_i'$ can be constructed in $O(|E^{\sigma}_i|\alpha(m,n))$ time.
\end{lemma}
\begin{proof} Let $(u,v)$ be an arbitrary edge in $E_i^{\sigma}$. 
We first consider the case where $(u,v)\in S_i$. Then, there is an edge $(r(u),r(v)) \in R_i$. By \Cref{thm:HalperinZwick}, there is a path $P$ between $r(u)$ and $r(v)$ in $(V(R_i),E'_i)$ that contains at most $2k-1$ edges. We write $P = (r(x_0)=r(u), (r(x_0),r(x_1)), r(x_1), (r(x_1),r(x_2)), \ldots, r(x_{p}) = r(v))$ as an alternating sequence of representative vertices and edges, where $x_0 = u, x_p = v$ and $p \le 2k-1$. Let $(y^2_\ell, y^1_{\ell+1})$ be the source edge in $S_i'$ that corresponds to the representative edge  $ (r(x_\ell),r(x_{\ell+1}))$, for each $\ell\in [0,p-1]$. Denote by $C_\ell$  the level-$i$ cluster with $r(C_\ell) = r(x_\ell)$. Let $y^1_0 = u$ and $y^2_{p}=v$. Let
	\begin{equation}
		Q = Q_{0}(y^1_0,y^2_0)\circ (y^2_0,y^1_1)\circ Q_1(y^1_1,y^2_1) \ldots \circ (y^2_{p-1},y^1_p)\circ  Q_p(y^1_p,y^2_p)
	\end{equation}
be a path from $u$ to $v$, where $Q_{\ell}(y^1_\ell,y^2_\ell)$ is a shortest path between $y^1_\ell$ and $y^2_\ell$ in $H^{\sigma}_{\leq i-1}[C_\ell]$, for each $0\leq \ell\leq p$, and $\circ$ is the path concatenation operator. By property \hyperlink{P2}{(P2)}, $w(Q_{\ell}(y^1_\ell,y^2_\ell)) \le 
g L_{i-1} = g \eps L_i$. It follows that 
	\begin{equation}\label{eq:weightQ}
		\begin{split}
			w(Q) &\leq (2k-1)L_i + (2k)g\eps L_i \leq (2k-1)(1+ 2g\eps)L_i\\
			&\leq (2k-1)(1+2g\eps)(1+\eps)w(u,v) \qquad \mbox{(since $w(u,v)\geq L_i/(1+\eps)$)}\\
			&\leq (2k-1)(1+(4g+1)\eps)w(u,v) \qquad \mbox{(since $\eps \leq 1$)}
		\end{split}
	\end{equation}
	
Thus, the stretch of $(u,v)$ is at most $(2k-1)(1+(4g+1)\eps)$. 
	
Next, we consider the complementary case that $(u,v) \nin S_i$. 
By definition, the edge $(u,v)$ is not in $S_i$ either because it is a self-loop or it is parallel to another edge $(u',v')$ that belongs to $S_i$, with $w(u',v') \le w(u,v)$. In the former case, property \hyperlink{P2}{(P2)} implies the existence of a path from $u$ to $v$ in $H^{\sigma}_{\leq i-1}$ of weight at most $gL_{i-1} ~=~ g\eps L_i ~\leq \frac{L_i}{1+\eps}~\leq~w(u,v)$ when $\eps <  \frac{1}{2g}$. Thus, in this case the stretch of edge $(u,v)$ is $1$.
For the latter case, let $C_u$ and $C_v$ be the level-$i$ clusters containing $u$ and $v$, respectively, and without loss of generality assume that $u' \in C_u$  and $v' \in C_v$. 
By property (\hyperlink{P2}{P2}),  $\dm(H^{\sigma}_{\leq i-1}[C_u]),\dm(H^{\sigma}_{\leq i-1}[C_v])) \le 
gL_{i-1} = g\eps L_i$. 
The same argument used for deriving \Cref{eq:weightQ}, when applied to the edge $(u',v')$ rather than $(u,v)$, yields:
\begin{equation} \label{eq:simple}
d_{H_{\leq i}}(u',v') ~\le~ (2k-1)(1+(4g+1)\eps)w(u',v') ~\le~  (2k-1)(1+(4g+1)\eps)w(u,v).
\end{equation}
By the triangle inequality, 
\begin{equation*}
	\begin{split}
		d_{H_{\leq i}}(u,v) &\leq~ 	d_{H_{\leq i}}(u',v') + \dm(H_{\leq i-1}[C_u]) + \dm(H_{\leq i-1}[C_v])\\
		&\leq~  (2k-1)(1+ (4g+1)\eps)w(u,v)  +  2g\epsi L_i  \qquad \mbox{(by \Cref{eq:simple})}\\
		&\leq~ (2k-1)(1+ (4g+1)\eps)w(u,v) + 4g\eps w(u,v) \qquad \mbox{(since $w(u,v) \geq L_i/(1+\eps) \geq L_i/2$)}\\
		&=~  (2k-1)(1+ (8g+1)\eps)w(u,v) \qquad \mbox{(since $k\geq 1$)}
	\end{split}
\end{equation*}
Thus, the stretch of edge $(u,v)$ is $(2k-1)(1+ (8g+1)\eps) = (2k-1)(1+O(\eps))$.
Summarizing, we have shown that in all cases the stretch of edge $(u,v)$ is at most $(2k-1)(1+O(\eps))$, as required.

By construction of the representative graph $R_i$, all clusters corresponding to vertices of $R_i$ are non-isolated, hence no vertex of $R_i$ is isolated, yielding
$|V(R_i)| = O(|E(R_i)|) = O(|E^{\sigma}_i|)$. 
Thus, the construction of the edge set $S_i$ and the representative graph $R_i$, via the usage of the $\djset$ data structure,  takes total time of $O(|E^{\sigma}_i|\alpha(m,n))$. The set of edges $S_i'$ by \Cref{thm:HalperinZwick} can be constructed in time $O(|E(R_i)| + |V(R_i)|) = O(|E^{\sigma}_i|)$ time. Thus, the total running time to construct $S_i'$ is  $O(|E^{\sigma}_i|\alpha(m,n))$.
\qed
\end{proof}

\paragraph{Constructing $\mathcal{C}_{i+1}$.~} Every cluster $C \in \mathcal{C}_i\setminus \mx$ becomes a level-$(i+1)$ cluster. We next focus on the level-$i$ clusters of $\mx$. Recall that $V(R_i)$ is the set of all representatives of clusters in $\mx$. We construct a collection $\mathcal{U}$ of vertex-disjoint subgraphs  of $R_i$ in the following two steps: 
\begin{itemize}
	\item[(1)] Initially, we greedily construct a maximal set of vertex-disjoint stars of $R_i$, and initialize $\mathcal{U}$ as this edge set; thus, each subgraph $U \in \mathcal{U}$ contains a vertex and all of its neighbors in $R_i$. 
	\item[(2)] We scan the remaining vertices in $V(R_i)$ that haven't been grouped to any subgraph in $\mathcal{U}$. 
	For every such remaining vertex $v \in V(R_i)$, it must have at least one neighbor that is contained in a subgraph $U\in \mathcal{U}$ (by the maximality of $\mathcal{U}$); we add to $U$ the vertex $v$ and an edge $(v,u)$ leading to such a neighbor $u$ of $v$ 
	 (chosen arbitrarily if there are multiple such neighbors). 
\end{itemize}  

For each subgraph $U$ in the resulting edge set $\mathcal{U}$, we form a level-$(i+1)$ cluster $C_U\in \mathcal{C}_i$ by taking the union of all the clusters whose representatives are $V(U)$ as $C_U$.

\begin{lemma}\label{lm:Ci1}All clusters in $\mathcal{C}_{i+1}$ satisfy Properties (\hyperlink{P1}{P1}) and (\hyperlink{P2}{P2}) when $\eps \leq 1/(2g)$ and $g \geq 9$, and they can be constructed in time $O(|E^{\sigma}_i| \cdot \alpha(m,n))$. Furthermore, every cluster $C_U\in \mathcal{C}_{i+1}$ that is formed from a subgraph $U \in \mathcal{U}$, as described above, is the union of at least $2$ level-$i$ clusters. 
\end{lemma}
\begin{proof} 
Property  (\hyperlink{P1}{P1}) holds trivially. 
To prove that  Property  (\hyperlink{P2}{P2}) holds, we first note that each subgraph $U \in \mathcal{U}$ (with vertices in $R_i$)
has hop diameter at most $4$,
which follows directly from the above two-step construction of $\mathcal{U}$.
Any edge connecting two vertices in $U$ corresponds to a source edge in $S_i$, and thus also in $E^{\sigma}_i$, and as such has length at most $L_i$, which implies that $C_U$ induces a subgraph of diameter at most $ 5(g L_{i-1}) + 4 L_i  = 5g\eps L_i + 4L_i \leq 9L_i = gL_i$, since $\eps\leq 1/g$ and $g \ge 9$. Thus, Property (\hyperlink{P2}{P2}) holds.

 The construction of the edge set $S_i$ and the representative graph $R_i$ takes total time of $O(|E^{\sigma}_i|\alpha(m,n))$ using the $\djset$ data structure. As for the construction of
the collection $\mathcal{U}$ of vertex-disjoint subgraphs  of $R_i$,
Step (1) of this construction, i.e.,  which constructs a maximal set of vertex-disjoint stars,
involves a greedy linear-time algorithm, whereas Step (2) naively takes linear time, so together they are implemented within time $O(|E(R_i)|) = O(|E^{\sigma}_i|)$. 
Constructing the corresponding clusters $\{C_U: U\in \mathcal{U}\}$ can be implemented within the same amount of time in the obvious way. The construction of clusters in $\mathcal{C}_{i+1}$ that are clusters in $\mathcal{C}_i\setminus \mx$ requires no extra time.

Finally, 
we argue that any cluster $C_U\in \mathcal{C}_{i+1}$ that is formed from a subgraph $U \in \mathcal{U}$ 
 contains at least $2$ level-$i$ clusters.
Indeed, any cluster formed in Step (1) of the construction of $\mathcal{U}$ contains at least 2 level-$i$ clusters,
by the maximality of $\mathcal{U}$ and since no vertex in $R_i$ is isolated.
Any remaining level-$i$ cluster must be grouped in Step (2) of the construction of $\mathcal{U}$ 
to clusters formed in Step (1), 
and this too holds by the maximality in Step 1 of the construction of $\mathcal{U}$ 
and since no vertex in $R_i$ is isolated.  \qed
\end{proof}

We are now ready to prove the first item of \Cref{thm:1}.

\begin{proof}[Proof of the first item of~\Cref{thm:1}] Recall that $H = \cup_{1\leq \sigma \leq \mu_\epsilon}H^{\sigma}$. Let $\Delta_{i+1} = |\mathcal{C}_i| - |\mathcal{C}_{i+1}|$. Recall that $\mathcal{C}_0$ is the set of $n$ singletons, i.e.,  $|\mathcal{C}_0| = n$. Thus, $\sum_{i\geq 0} \Delta_{i+1} \leq |\mathcal{C}_0|  = n$. 
	
	By \Cref{lm:Ci1}, $\Delta_{i+1} \geq \frac{|V(R_i)|}{2}$. Furthermore, \Cref{thm:HalperinZwick} yields $|S_i'| = O(|V(R_i)|^{1+1/k})$, hence $|S_i'| = O(n^{1/k}) \cdot \Delta_{i+1}$. Thus, we have:
	\begin{equation}\label{eq:EHsigma}
		|E(H^{\sigma})| ~=~ |\cup_{i\geq 0} E(H^{\sigma}_i)| ~=~ 
		\sum_{i\geq 0} |S_i'| ~=~ \sum_{i\geq 0}O(n^{1/k}) \cdot \Delta_{i+1} ~=~ O(n^{1+1/k}).
	\end{equation}

	 The sparsity of $H$ is $O(n^{1/k} \cdot \frac{\log(1/\eps)}{\eps})$ by \Cref{eq:EHsigma}. The stretch of $H$ is at most $(2k-1)(1+O(\eps))$ by~\Cref{lm:Stretch}; we can reduce the stretch down to $(2k-1)(1+\eps)$ by scaling $\eps \leftarrow \eps/c$, for a sufficiently large constant $c$, which will affect the sparsity and runtime bounds by constant factors. 	The time needed to construct $H^\sigma$ is $O(\sum_{i\geq 0}|E^{\sigma}_i| \cdot \alpha(m,n)) = O(m \cdot \alpha(m,n))$ by \Cref{lm:Stretch} and \Cref{lm:Ci1}. Thus, the overall time needed to construct $H$, when also considering the 
	runtime $O(\srt(m))$ for computing the partition of $E$ into the sets $\{E^{\sigma}\}_{1\leq \sigma \leq \mu_\epsilon}$,
	is $O(m \cdot \alpha(m,n) \cdot \frac{\log(1/\eps)}{\eps} + \srt(m))$. \qed
\end{proof}

\section{ A Linear Time Algorithm in the \wrdram Model}\label{sec:sparseRAM}

In this section, we prove the second item of \Cref{thm:1}. We follow the same framework as in \Cref{sec:PointerMachine}; our focus, as before, is on constructing a $(2k-1)(1+O(\eps))$-spanner $H^{\sigma}$ for $E^{\sigma}$, for a fixed $\sigma \in [1,\mu_{\eps}]$.  The construction is carried out in levels, where $H_i^{\sigma}$ is constructed at level $i$, and uses a hierarchy of clusters such that each cluster $C\in \mathcal{C}_i$ satisfies two properties that are similar to those used in \Cref{sec:PointerMachine}, namely Properties (\hyperlink{P1}{P1}) and (\hyperlink{P2}{P2}).

We also use a \djset~data structure to represent clusters in $\mathcal{C}_i$. However, our construction relies on a special case of \djset~, where the set of \union~operations are pre-specified at the outset of the construction. Gabow and Tarjan~\cite{GT85} designed a data structure for this special case of \djset~in the \wrdram model; this result is summarized in the following theorem.

\begin{theorem}[Gabow and Tarjan~\cite{GT85}]\label{thm:GabowTarjan} 
Let $T$ be a rooted tree with $n$ vertices. One can design a \djset~data structure in the \wrdram model that maintains disjoint sets of $V(T)$ and supports $m$ $\union$ and $\find$ operations in $O(m+n)$ total time, in which each \union~operation is of the form $\union(v,p_T(v))$ for some non-root vertex $v \in V(T)$. Here $p_T(v)$ denotes the parent of  $v$ in $T$.  
\end{theorem}

We emphasize that the \djset~data structure of Gabow and Tarjan in \Cref{thm:GabowTarjan} only works in the \wrdram model.  The tree $T$ in \Cref{thm:GabowTarjan} is called a \emph{union tree} of the \djset~data structure. We use $\link(v)$ to specifically denote the \union~operation of the form $\union(v, p_T(v))$. 

The construction of~\Cref{sec:PointerMachine} achieves a super-linear running time.
To improve this runtime to linear in $m$, we plug the following new ideas on top of the construction of~\Cref{sec:PointerMachine}.

The second term in the super-linear runtime $O(m \cdot \alpha(m,n) \cdot \frac{\log(1/\eps)}{\eps}+\srt(m))$, namely $\srt(m)$,
stems from the time needed to compute the partition of $E$ into the sets $\{E^{\sigma}\}_{1\leq \sigma \leq \mu_\epsilon}$,
which boils down to sorting the indices of the non-empty sets in $\{E^{\sigma}\}_{1\leq \sigma \leq \mu_\epsilon}$. 
In the $\rdrm$ model, we employ a rather simple trick to carry out such an index sorting in time $O(m \cdot \frac{\log(1/\eps)}{\eps})$; the details of this optimization appear in~\Cref{fusion}.

The main obstacle lies in shaving the  factor $\alpha(m,n)$ from the first term $O(m \cdot \alpha(m,n) \cdot \frac{\log(1/\eps)}{\eps})$. 
For this optimization, the two key ideas are the following:

\begin{itemize}
	\item \textbf{Idea 1.~} We use an MST for $G$ as the union tree for the \djset~data structure.
	In the \wrdram model, Fredman and Willard~\cite{FW94} designed an algorithm to construct a minimum spanning tree  in $O(m)$ time. 
	Let	$\mst$ be an arbitrary minimum spanning tree for $G$; we root $\mst$ at an arbitrary vertex $r$. 
	
	\item  \textbf{Idea 2.~} We guarantee that every level-$i$ cluster $C \in \mathcal{C}_i$ \emph{induces} a subtree of $\mst$ of diameter at most $gL_{i-1}$, for some constant $g$. As we will show in the sequel, by forcing clusters to induce subtrees of $\mst$, we are able to use \link~operations to form level-$(i+1)$ clusters from level-$i$ clusters, which is the source of our speed-up. 
	The crux of our construction is in realizing idea 2.  
\end{itemize}
\Cref{thm:GabowTarjan} guarantees that each of the $\union$ and $\find$ operations takes $O(1)$ amortized time. As a result, we shave the $\alpha(m,n)$ factor in the running time of the algorithm from \Cref{sec:PointerMachine}.

Next we proceed to the details of the linear-time construction. 
The construction will satisfy the following two properties of clusters in $\mathcal{C}_i$,
the first of which is identical to Property (\hyperlink{P1}{P1}) of~\Cref{sec:PointerMachine} whereas the second is an adaptation of Property  (\hyperlink{P2}{P2}).

\begin{itemize}[noitemsep]
\item \textbf{(P1')~} 	\hypertarget{P1'}{}
	Each cluster $C\in \mathcal{C}_i$ is a subset of $V$. Furthermore, clusters in $\mathcal{C}_i$ induce a partition of $V$.
\item \textbf{(P2')~} 	\hypertarget{P2'}{} Each cluster $C\in \mathcal{C}_i$ induces a (connected) subtree $\mst[C]$ of $\mst$ with diameter at most $gL_{i-1}$, for some constant $g$ (the same constant used in Idea 2 above which is different than the one used
in  (\hyperlink{P2}{P2})).  
\end{itemize}

We will add all edges of $\mst$ to the   spanner, by setting $H^{\sigma}_{0}$ as  $\mst$,
which adds one unit to the sparsity and lightness.
Property (\hyperlink{P2'}{P2'}) is inherently more restrictive than
 Property (\hyperlink{P2}{P2}), as it aims at guaranteeing the same (perhaps up to a constant factor) diameter bound, but when restricted to subtrees of $\mst$. 

\paragraph{Representing $\mathcal{C}_i$.~} As in~\Cref{sec:PointerMachine}, we use the \djset~data structure to represent clusters in $\mathcal{C}_i$, but  we use the  data structure   provided by~\Cref{thm:GabowTarjan}, which guarantees constant amortized cost.
As a result, we will maintain the property that the representative $r(C)$ of any cluster $C \in \mathcal{C}_i$ is always set to be the \emph{root of the subtree}  $\mst[C]$. By setting the representative of a cluster $C$ to be its root, $C$ can be united with other clusters via $\link(r(C))$, which is crucial for applying the result of~\Cref{thm:GabowTarjan}. The children of $C$ can be united to $C$ by the same way. 

\paragraph{Constructing $H^{\sigma}_i$.~} The construction is the same as the construction of $H^{\sigma}_i$ in \Cref{sec:PointerMachine}. Specifically, we construct a set of level-$i$ clusters $\mx$, the representative graph $R_i$, and the edge set $S'_i$, which is obtained by running the spanner algorithm of~\Cref{thm:HalperinZwick} to $R_i$. Since the \union~and \find~operations now admit $O(1)$ (amortized) time, we derive the following lemma, whose proof follows along similar lines as those in the proof of~\Cref{lm:Stretch}.

\begin{lemma}\label{lm:RAM-Stretch} $d_{H_{\leq i}}(u,v) \leq (2k-1)(1+O(\eps))w(u,v)$ for every edge $(u,v) \in E^{\sigma}_i$, assuming $\eps < 1/(2g)$. Furthermore, $S'_i$ can be constructed in $O(|E^{\sigma}_i|)$ time.
\end{lemma}

\paragraph{Constructing $\mathcal{C}_{i+1}$.~} Our construction of $\mathcal{C}_{i+1}$ relies on the notion of {\em cluster forest} defined below; see \Cref{fig:ClusterForest} for an illustration.

\begin{definition}[Cluster Forest]\label{def:cluster_tree} Let $\mathcal{Y}\subseteq \mathcal{C}_i$ be a set of level-$i$ clusters. A {\em cluster forest} for $\mathcal{Y}$, denoted by $\mathcal{F}_{\mathcal{Y}}$, is a directed forest with a weight function $\omega$ on the edges such that:
	\begin{itemize}
		\item[(1)] Each node $\varphi_C \in \mathcal{F}_{\mathcal{Y}}$ corresponds to a cluster $C\in \mathcal{Y}$,
		\item[(2)] There is a directed edge $(\varphi_{C_1} \rightarrow \varphi_{C_2})$ in the forest $\mathcal{F}_{\mathcal{Y}}$ if $C_2$ contains the parent, say $p_{\mst}(v)$, of the representative, say $v$, of $C_1$. Furthermore, $\omega(\varphi_{C_1} \rightarrow \varphi_{C_2}) = w(v,p_{\mst}(v))$.
	\end{itemize}
\end{definition}

\begin{figure}[!h]
	\begin{center}
		\includegraphics[width=0.9\textwidth]{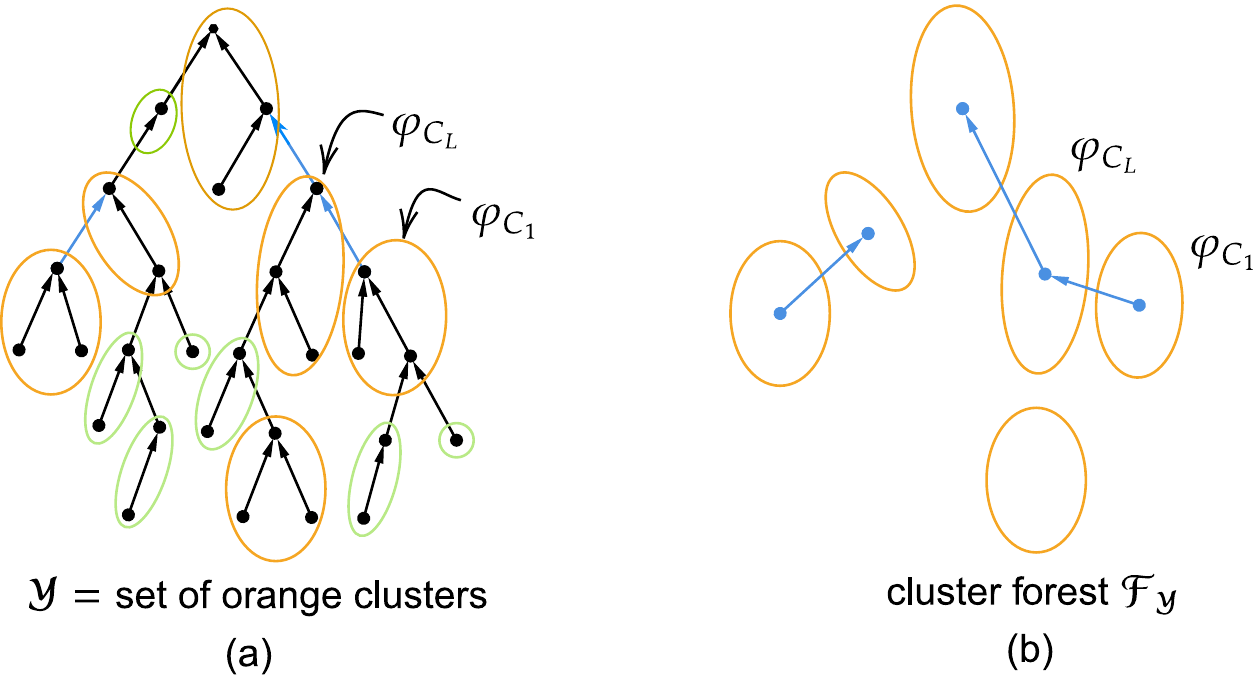}
	\end{center}
	\caption{(a) Level-$i$ clusters induce subtrees of $\mst$ enclosed by oval curve, and (b) a cluster forest $\mf_{\my}$.}
	\label{fig:ClusterForest}
\end{figure}

By definition, every edge of a cluster forest  $\mathcal{F}_{\mathcal{Y}}$ corresponds to an $\mst$ edge.  Let $\calmst_i = \mathcal{F}_{\mathcal{C}_i}$ be the cluster forest defined for the entire set $\mathcal{C}_i$ of level-$i$ clusters;  by Property (\hyperlink{P2'}{P2'}), it holds that $\calmst_i$  is a tree. We stress that $\calmst_i$ is only used in the analysis of our algorithm; indeed, computing $\calmst_i$, at least naively, would require $\Omega(|\mc_i|)$ time, which is too costly. 

For a set $\mathcal{Y}$ of level-$i$ clusters, 
we say that the cluster forest $\mathcal{F}_{\mathcal{Y}}$  is \emph{$L_i$-bounded} if every edge in it has weight at most $L_i$. The following lemma is the crux of our construction.

Recall that $\mx$ denotes the set of  all non-isolated level-$i$ clusters in the representative graph $R_i$.

\begin{lemma}\label{lm:XTree} Let $\mathcal{A}_i$ be the set of $\calmst_i$ edges of weight at most $L_i$, and let $\my$ be the set of nodes that are incident on at least one edge in $\mathcal{A}_i$. Let $\mathcal{F}_{\mathcal{Y}}$ be the forest with node set $\my$ and edge set $\mathcal{A}_i$.  Then the following two conditions hold: 
	\begin{enumerate}
		\item[(1)] $\mx \subseteq \my$.
		\item[(2)] Every tree in $\mathcal{F}^{\prune}_{\mathcal{Y}}$  has at least 2 nodes.
	\end{enumerate}
\end{lemma}
\begin{proof}
	Condition (2) follows directly from the construction. We next prove that Condition (1) holds. 

 Let $\varphi_{C_u}$ be the node corresponding to a level-$i$ cluster $C_u$ in $\mathcal{X}$. By the definition of $\mx$, there is an edge $(u,v) \in E^{\sigma}_i$ such that $u \in \varphi_{C_u}$. Let $C_v$ be the level-$i$ cluster containing $v$ and $\varphi_{C_v}$ be the node corresponding to $C_v$. If $(u,v) \in \mst$, then $\varphi_{C_u} \in \my$, and we're done. 
	
We henceforth assume that $(u,v) \nin \mst$. Consider the fundamental cycle $C_{uv}$ of $\mst$ formed by $\mst[u,v]$ and edge $(u,v)$. By the cycle property of $\mst$, every edge $e \in \mst[u,v]$ satisfies $w(e)\leq w(u,v)$. Recall that
 $\calmst_i$  is a tree by Property (\hyperlink{P2'}{P2'}). Moreover, by the definition of $\calmst_i$,  each edge in $\calmst_i[\varphi_{C_u}, \varphi_{C_v}]$ corresponds to an edge in $\mst[u,v]$, and so has weight at most $w(u,v)\leq L_i$. Hence $\varphi_{C_u}$ is incident to an edge of $\mathcal{A}_i$ by the definition of $\mathcal{A}_i$, which yields $\varphi_{C_u} \in \my$. \qed
\end{proof}

We now construct the set of level-$(i+1)$ clusters $\mathcal{C}_{i+1}$ as follows. Let $\mf_{\my}$ be the cluster forest for $\my$ provided by~\Cref{lm:XTree}. We construct $\mathcal{C}_{i+1}$ as follows. Every level-$i$ cluster $C \in \mathcal{C}_i\setminus \my$ becomes a level-$(i+1)$ cluster. 
Then, we   construct a collection $\mathbb{U}$ of subtrees of  $\mathcal{F}_{\my}$, such that each subtree $\mathcal{U}\in \mathbb{U}$ contains at least two nodes and has hop-diameter at most $4$. 
 For each subtree $\mathcal{U}$, we form a level $(i+1)$ cluster $C_{\mathcal{U}} = \cup_{\varphi_C \in \mv(\mU)}C$. 
We note that $\mathbb{U}$ can be constructed greedily via the same algorithm used in \Cref{sec:PointerMachine}, within time  $O(|\my|)$.

In the following lemma we assume that the set of clusters $\my$ is given to us. In the proof of \Cref{thm:1} where we use \Cref{lm:RamCi1}, we will specify the construction of $\my$.

\begin{lemma}\label{lm:RamCi1}
All clusters in $\mathcal{C}_{i+1}$ satisfy Properties (\hyperlink{P1'}{P1'}) and (\hyperlink{P2'}{P2'})
when $\eps \leq 1/(2g)$ and $g \geq 9$, and they can be constructed in time $O(|\my|)$. Furthermore, every cluster $C_\mU\in \mathcal{C}_{i+1}$ that is formed from a subgraph $\mU \in \mathbb{U}$, as described above, is the union of at least $2$ level-$i$ clusters.
\end{lemma}
\begin{proof}
The proof of this lemma follows similar lines to those in the proof of~\Cref{lm:Ci1} from Section~\Cref{sec:PointerMachine}, hence we aim for conciseness. As mentioned, $\mathbb{U}$ can be constructed  within time  $O(|\my|)$.
	
	Recall that every edge in $\mf_\my$ corresponds to an edge of the form $v\rightarrow p_\mst(v)$ for some vertex $v\in V$. Thus, for each subgraph $\mU\in \mathbb{U}$, the level-$(i+1)$ cluster $C_{\mU}$ can be constructed by calling $|\mv(\mU)|-1$ \link~operations. Therefore, $\{C_\mU: \mU\in \mathbb{U}\}$ can be constructed in time $O(\sum_{\mU\in \mathbb{U}}|\mU|) = O(|\my|)$. Note that we do not pay any running time for constructing clusters in $\mathcal{C}_{i+1}$ that are clusters in $\mathcal{C}_i\setminus \my$. Therefore $\mathcal{C}_{i+1}$ can be constructed in $O(|\my|)$ time.
	
Property  (\hyperlink{P1'}{P1'}) holds trivially. 	Property (\hyperlink{P2'}{P2'}) follows from the fact that each subgraph $\mU \in \mathbb{U}$ has hop diameter at most $4$ and that each edge between two nodes in $\mU$ corresponds to an edge in $\mst$ of length at most $L_i$ since every edge of $\mf_{\my}$ has a weight at most $L_i$ by construction.

Note that $\mathbb{U}$ is constructed using same two-step algorithm used in \Cref{sec:PointerMachine}. Thus, the same argument in \Cref{lm:Ci1} applies to this case. Specifically, any cluster formed in Step (1) of the construction of $\mathbb{U}$ contains at least 2 nodes,  since no vertex in $\mf_\my$ is isolated, and any remaining  node must be grouped in Step (2) of the construction of $\mathbb{U}$ to sugraphs formed in Step (1). \qed
\end{proof}

We are now ready to prove the second item of \Cref{thm:1}.

\begin{proof}[Proof of the second item of~\Cref{thm:1}] Recall that $H = \cup_{1\leq \sigma \leq \mu_\sigma}H^{\sigma}$.	 We employ a similar charging argument to the one used in~\Cref{sec:PointerMachine} to bound $|E(H^{\sigma})|$.	Let $\Delta_{i+1} =  |\mathcal{C}_i| -  |\mathcal{C}_{i+1}|$. Note that $|\mathcal{C}_0| = n$, hence $\sum_{i\geq 0} \Delta_{i+1} \leq n$.  By \Cref{lm:RamCi1} and \Cref{lm:XTree}, we have   $\Delta_{i+1} \geq \frac{|\my|}{2}  = \Omega(|\mx|) = \Omega(|V(R_i)|)$. (Note that $|\mx| = |V(R_i)|$.) Thus, \Cref{eq:EHsigma} of~\Cref{sec:PointerMachine} holds in this case as well. It follows that the sparsity of $H$ is $O(n^{1/k} \cdot \frac{\log(1/\eps)}{\eps})$. The stretch is $(2k-1)(1+O(\eps))$ by~\Cref{lm:RAM-Stretch};
we can reduce the stretch down to $(2k-1)(1+\eps)$ by scaling $\eps \leftarrow \eps/c$, for a sufficiently large constant $c$, which will affect the sparsity and runtime bounds by constant factors.
	The runtime to construct $H^\sigma$ is $O(\sum_{i\geq 0}|E^{\sigma}_i|) = O(m)$ by \Cref{lm:RAM-Stretch}.
	
		We now bound the time to construct the clusters in $\mathcal{C}_{i+1}$. The main difficulty is that the size of $\my$ constructed in \Cref{lm:XTree} could be much larger than $|E^{\sigma}_i|$, hence we cannot bound the runtime by $|E^{\sigma}_i|$
		as we did in \Cref{sec:PointerMachine}. Here we employ a more delicate argument. At the outset of the construction, we divide the edges of $\mst$ into levels as we did for $E^{\sigma}_i$. 
		The level-$i$ edges of $\mst$, denoted by $B_i$, include every edge of length larger than $L_{i-1}$ and at most $L_i$. The time to construct $B_i$ is $O(n\log(1/\eps))$, following the same index-sorting argument used for constructing $E^{\sigma}_i$ efficiently in  \Cref{fusion}.

	At the outset of the construction of $\mathcal{C}_{i+1}$, we assume that we are given the set of edges $\mathcal{D}_{i-1}$ that contains every edge of weight at most $L_{i-1}$ of $\calmst_i$. For level $i = 0$, we set $\mathcal{D}_{i-1} = \emptyset$. Let $\mathcal{B}_i$ be the set of edges of $\calmst_i$ corresponding to edges in $B_i$. The edge set $\mathcal{B}_i$ can be constructed in $O(|B_i|)$ time as follows. For each edge $(u,v)\in B_i$, we add an edge $(\varphi_{C_u}, \varphi_{C_v})$ to $\mathcal{B}_i$, where $C_u$ and $C_v$ are the two level-$i$ clusters containing $u$ and $v$, respectively, which can be found via $\find(u)$ and $\find(v)$. 
	
	The set of edges $\mathcal{A}_i$ defined in \Cref{lm:XTree} is $\mathcal{D}_{i-1}\cup \mathcal{B}_i$. Note that  $|\mathcal{A}_i| \leq |\my|$ since $\mathcal{F}_{\mathcal{Y}}$ is acyclic. Thus, the running time to construct $\mathcal{A}_{i}$ is  $O(|\my|)$, as we have both $\mathcal{D}_{i-1}$ and $\mathcal{B}_i$ stored in a list data structure. To construct the set of $\mathcal{D}_i$ for the construction at the next level, we simply identify edges in  $\mathcal{F}_{\mathcal{Y}}$ that are between two different subgraphs  $\mathbb{U}$ in the construction of $\mathcal{C}_{i+1}$. Thus, the running time to construct $\mathcal{D}_i$ is also $O(|\my|)$. The running time to construct $\mathcal{C}_{i+1}$ is $O(|\my|)$ by \Cref{lm:RamCi1}. It follows that the total running time of the construction of clusters at level $i$ is  $O(|\my|)$.   Since $\Delta_{i+1} \geq \frac{|\my|}{2}$, the time to construct $\mathcal{C}_{i+1}$ is bounded by $O(\Delta_{i+1})$, where $\Delta_{i+1} = |\mathcal{C}_i| - |\mathcal{C}_{i+1}|$. It follows that the total running time to construct clusters over all levels is $\sum_{i\geq 0} O(\Delta_{i+1})  = O(n)$. 
	
	In summary, the running time to construct $H$ is $O((m+n) \cdot \frac{\log(1/\eps)}{\eps}) = O(m \cdot \frac{\log(1/\eps)}{\eps})$.\qed
\end{proof}

\subsection{Index sorting in linear time} \label{fusion}

First, we assume that the word size is $\bar{w} \ge \log(n)$ and all edge weights are bounded above by $2^{\bar{w}}$, as per the \rdrm model.  The total number of different indices is given by $\log_{1+\eps} 2^{\bar{w}} = \Theta(\bar{w}/\eps)$. It follows that the number of integers is $ n' \leq \bar{w}/\eps$.  In this range of values, predecessor search can be done in $O(\log(n’)/\log \bar{w}) = O(\log(1/\eps))$ time using the fusion tree data structure \cite{FW90} (see also~\cite{PT06}). 
Consequently, the time needed to compute the partition of $E$ into the sets $\{E^{\sigma}\}_{1\leq \sigma \leq \mu_\epsilon}$,
which involves index sorting via predecssor search, is bounded by $ O(m\log(1/\eps))$.
Partitioning the set of edges of $\mst$ into levels can be done in the same way; the running time is  $O(n\log(1/\eps))$ as there are $n-1$ edges in $\mst$.
Summarizing, the running time of these partitioning steps is bounded by  $O((m+n)\log(1/\eps)) =  O(m\log(1/\eps))$.

\section{Optimally Sparse and Light Spanners in $O(m\alpha(m,n))$ Time}

Le and Solomon~\cite{LS21} recently show that a $(2k-1)(1+\eps)$-spanner with lightness $O_{\eps}(n^{1/k})$ can be constructed in $O_{\eps}(m\alpha(m,n))$ time; the notation $O_{\eps}(.)$ hides a polynomial factor of $1/\eps$.  However, their spanner is not sparse, i.e., in the worst case, the number of edges of the spanner is $\Omega(m)$, which could be $\Omega(n^{2})$ for dense graphs. Here we use the insights we develop in \Cref{sec:PointerMachine} and \Cref{sec:sparseRAM} on top of the construction of \cite{LS21} to obtain a  $(2k-1)(1+\eps)$-spanner that is both sparse and light as claimed in \Cref{thm:2}.

First, we briefly recap the algorithm of Le and Solomon~\cite{LS21}, called \emph{LS algorithm}. LS algorithm first divides $E$ into two sets of edges:  $E_{\light} = \{e \in : w(e) \leq \frac{w(\mst)}{m\eps}\}$ and $E_{\heavy} = E\setminus E_{\light}$. Every edge in $E_{\light}$ shall be added to the final spanner, and this only incurs an additive $+\frac{1}{\eps}$ in the lightness since:
\begin{equation}\label{eq:weightElight}
	w(E_{\light}) \leq m \cdot \frac{w(\mst)}{m\eps}  \leq \frac{w(\mst)}{\eps}.
\end{equation}
For edges in $E_{\heavy}$, LS algorithm constructs a  $(2k-1)(1+\eps)$-spanner $H_{\heavy}$  that has two properties: 
\begin{equation} \label{eq:Hprop}
	\begin{split}
		(1)\quad &w(H_{\heavy})\leq O_{\eps}(n^{1+1/k})w(\mst) \\
		(2) \quad &d_G(u,v)\leq d_{H_{\heavy}}(u,v)\leq (2k-1)(1+\eps)d_G(u,v) \quad \forall (u,v)\in E_{\heavy}
	\end{split}
\end{equation}
The final spanner of the graph is $H = H_{\heavy}\cup E_{\light}$.  By \Cref{eq:weightElight} and \Cref{eq:Hprop}, it follows that $w(H) = (O_{\eps}(n^{1/k}) + \frac{1}{\eps})w(\mst) = O_{\eps}(n^{1/k}) w(\mst)$, and hence the lightness of $H$ is $O_{\eps}(n^{1/k})$. 

Our first observation is that in LS algorithm, $H_{\heavy}$ does not contain any other edge of $E_{\light}$, except for $\mst$ edges. It follows that if we construct  a $(2k-1)(1+\eps)$-spanner  $H_{\light}$ for  the subgraph of $G$ induced by $E_{\light}\cup \mst$ by applying the construction in \Cref{thm:1}, and set $H = H_{\light} \cup H_{\heavy} \cup \mst$, then $H$ is still a $(2k-1)(1+\eps)$-spanner  of $G$. Furthermore, $w(H) \leq w(E_{\light}) + w(H_{\heavy}) + w(\mst) = O_{\eps}(n^{1/k})w(\mst)$. Thus, the lightness of $H$ is $O_{\eps}(n^{1/k})$.  Observe that $H_{\light}\cup \mst$ has sparsity $O_{\eps}(n^{1/k})$. It follows that, to guarantee that $H$ has sparsity $O_{\eps}(n^{1/k})$, we need to construct  $H_{\heavy}$  such that its sparsity and lightness are both $O_{\eps}(n^{1/k})$. Our construction crucially makes use of the cycle property of $\mst$ following the same spirit of the construction in \Cref{sec:sparseRAM}.

\subsection{The construction of $H_{\heavy}$}

We assume that $E_{\heavy}$ has no edges of weight at least $w(\mst)$ since we could safely remove them from $E_{\heavy}$ without affecting the stretch of the construction.  The spanner $H_{\heavy}$  constructed by LS algorithm is a subgraph of $G_{\heavy}$, which is a subgraph $G$ induced by $E_{\heavy}\cup E(\mst)$.  However, the construction operates on a graph $\tilde{G}$ obtained from  $G_{\heavy}$ by subdividing edges of $\mst$ using \emph{virtual vertices}. Specifically, we define $\bar{w} = w(\mst)/\eps$, and for each edge of $e\in \mst$, if $w(e)\geq \bar{w}$, we subdivide $e$ into $\lceil \frac{w(e)}{\bar{w}} \rceil$ edges, each of weight at most $\bar{w}$, whose total weight is $w(e)$.  Let $\msttilde$  be the subdivided $\mst$ and $\tilde{G} = (\tilde{V}, E(\msttilde)\cup E_{\heavy})$.  That is, $\tilde{G}$ and $G$ share the same set of edges $E_{\heavy}$. Vertices in $\tilde{V}\setminus V$ are called \emph{virtual vertices}. Le and Solomon~\cite{LS21} observed that:

\begin{observation}[Observation 3.4 in \cite{LS21}]\label{obs:virtualNum} $|\tilde{V}| = O(m)$. 
\end{observation}

We now divide $E_{\heavy}$ further into subsets $\{E^{\sigma}\}_{1\leq \sigma \leq \mu_\eps}$ with $\mu_\eps = O(\frac{1}{\eps}\log(1/\eps))$, as we did in \Cref{sec:PointerMachine} (\Cref{eq:EsigmaI}):

\begin{equation}\label{eq:Esigmaixdef}
	E^{\sigma}_i = \left\{e : \frac{L_i}{1+\eps} \leq w(e) \leq L_i \right\} \mbox{ with } L_i = L_{0}/\eps^i, L_0 = (1+\eps)^{\sigma}\bar{w}~. 
\end{equation}

We note that constructing  every $E^{\sigma}_i$ can be done in $O(m)$ by simply sorting all indices $i$ such that $E^{\sigma}_i \not=\emptyset$. This is because the maximum index $i_{\max}$ is $O(\log(n))$ (Claim 3.5 in \cite{LS21}) and hence, the sorting step takes only $O(\log(n)\log\log(n)) = O(n)$ time.

The construction then focuses on each set $E^{\sigma}$ separately. That is, we construct a $(2k-1)(1+O(\eps))$-spanner $H^{\sigma}$ for each set $E^{\sigma}$ in $O(m\alpha(m,n))$ time, and set  $H_{\heavy} = \cup_{1\leq \sigma \leq \mu_\eps}H^{\sigma}$. It follows that the running time to construct $H_{\heavy}$ is $O(m\alpha(m,n)/\eps)$. Here we slightly abuse notation as $H^{\sigma}$ is a subgraph of $\tilde{G}$ instead of being a subgraph of $G_{\heavy}$. However, the difference between $\tilde{G}$ and $G_{\heavy}$ lies only in $\mst$ vs $\msttilde$, and by assuming that $H^{\sigma}$ contains $\msttilde$, we can transform $H^{\sigma}$ to a subgraph of $G_{\heavy}$ by replacing each path of subdividing virtual vertices with the corresponding original edge of $\mst$. 

For notational convenience, we set $H^{\sigma}_{0} = (\tilde{V}, E(\msttilde))$. The construction of $H^{\sigma}$ happens in \emph{levels}: at level $i$, we construct a subgraph $H^{\sigma}_i$ such that $H^{\sigma}_{\leq i}$ is a $(2k-1)(1+O(\eps))$-spanner for edges in $E^{\sigma}_{\leq i}$. Here $H^{\sigma}_{\leq i} = \cup_{0\leq j \leq i}H^{\sigma}_j$ and  $E^{\sigma}_{\leq i} = \cup_{0\leq j \leq i}  E^{\sigma}_{j}$. Recall that $E^{\sigma}_{0} = \emptyset$ since every edge in $E_{\heavy}$ has a weight of at least $\bar{w}/\eps$.  By induction, $H^{\sigma} \defi \cup_{i\geq 0}H^{\sigma}_i$ is a $(2k-1)(1+O(\eps))$-spanner for $\tilde{G}$. 

 Similar to the construction of a sparse spanner in \Cref{sec:PointerMachine}, we construct a hierarchy of clusters, and each level $i\geq 1$ of the construction is associated with a set of clusters $\mathcal{C}_i$ satisfying the following properties:
 
 \begin{enumerate}[noitemsep]
 	\item[(1)] \hypertarget{P1L}{} Each cluster $C\in \mathcal{C}_i$ is a subset of $V$. Furthermore, clusters in $\mathcal{C}_i$ induce a partition of $\tilde{V}$.
 	\item[(2)]\hypertarget{P2L}{} Each cluster $C\in \mathcal{C}_i$  is the union of $\Omega(1/\eps)$ clusters in $\mathcal{C}_{i-1}$ for any $i\geq 2$.  
 	\item[(3)]\hypertarget{P3L}{} Each cluster $C\in \mathcal{C}_i$ induces a subgraph $H_{\leq i-1}^{\sigma}[C]$ of diameter at most $gL_{i-1}$ for some constant $g$.
 \end{enumerate}
  
 By property \hyperlink{P3L}{(3)}, clusters at level $1$ are subgraphs of $H_0$, which is $\msttilde$. The construction is described in the following lemma.
 
 \begin{lemma}[Lemma 3.8 in \cite{LS21}]\label{lm:level1Const}  In time $O(m)$, we can construct a set of level-$1$ clusters $\mathcal{C}_1$ such that, for each cluster $C\in \mathcal{C}_1$, the subtree $\msttilde[C]$ 	of $\msttilde$ induced by $C$ satisfies $L_0 \leq \dm(\msttilde[C]) \leq 14L_0$. 
 \end{lemma}

Thus, by choosing $g\geq 14$, property \hyperlink{P3L}{(3)} is satisfied for clusters in $\mathcal{C}_1$. Property \hyperlink{P1L}{(1)} follows directly from \Cref{lm:level1Const}, and property \hyperlink{P2L}{(2)} is not applicable. 

A crucial component in LS algorithm is a potential function $\Phi: 2^{\tilde{V}}\rightarrow \mathbb{R}^+$ that associates each cluster $C\in \mathcal{C}_i$ with a potential value $\Phi(C)$. Let $\Phi_i = \sum_{C\in \mathcal{C}_i} \Phi(C)$ be the total potential value at level $i$. The potential values of level $1$ clusters are defined as follows:
\begin{equation}\label{eq:level1Potential}
	\Phi(C) = \dm(\msttilde[C]) \quad\forall C\in \mathcal{C}_1
\end{equation}
By \Cref{lm:level1Const}, we have that:
\begin{equation}\label{eq:level1TotalPotential}
	\Phi_1 = \sum_{C\in \mathcal{C}_1}\dm(\msttilde[C]) \leq w(\mst)
\end{equation}

Next, Le and Solomon~\cite{LS21} define a potential change $\Delta_{i+1} = \Phi_i - \Phi_{i+1}$. Let $i_{\max}$  be the maximum level, and define $\Phi_{i_{\max}+1} = 0$. The idea is to bound the weight of the to-be-constructed spanner $H^{\sigma}_i$  by the potential change $O_{\eps}(n^{1/k})\Delta_{i+1}$ (modulo a small additive term that we will describe later). It follows that we can bound the weight of $H^{\sigma}$, again modulo a small additive term, as follows.
\begin{equation}
	w(H^{\sigma}) \leq  O_{\eps}(n^{1/k}) \sum_{i=0}^{i_{\max}} \Delta_{i+1} = O_{\eps}(n^{1/k})(\Phi_1 - \Phi_{i_{\max}+1}) = O_{\eps}(n^{1/k}) \Phi_1 = O_{\eps}(n^{1/k}) w(\mst)
\end{equation}

In~\cite{LS21}, Le and Solomon showed the following lemma, which is the key to their construction.

\begin{lemma}[Lemma 2.6 and Theorem 1.9~\cite{LS21}]
	\label{lm:LSConstruction} For each level $i\geq 1$, there is an algorithm that can compute a subgraph $H^{\sigma}_i$ \emph{induced by a subset of $ E_i^{\sigma}$}, as well as the set of level-$(i+1)$ clusters $\mathcal{C}_{i+1}$ satisfying properties \hyperlink{P1L}{(1)}-\hyperlink{P3L}{(3)}  given a set of clusters $\mathcal{C}_i$ at level $i$,  such that:  
	\begin{enumerate}[noitemsep]
		\item[(1)] $w(H^{\sigma}_i) =  O_{\eps}(n^{1/k}) \Delta_{i+1} + a_i$ for some $a_i\geq 0$ such that $\sum_{1\leq i\leq i_{\max}} a_i = O_{\eps}(n^{1/k})w(\mst)$.
		\item[(2)] for every $(u,v)\in E^{\sigma}_i$, $d_{H^{\sigma}_{\leq i}}(u,v)\leq (2k-1)(1+(10g+1)\epsilon)w(u,v)$.
	\end{enumerate}
	Furthermore, the total running time of the construction of \emph{all levels} is $O(m\alpha(m,n))$  in the pointer-machine model.
\end{lemma}

In \Cref{lm:LSConstruction}, $a_i$ is a corrective term added to handle some edge cases where $\Delta_{i+1} = 0$ or even negative. The stretch is $(2k-1)(1+(10g+1)\epsilon)$ instead of $(2k-1)(1+\eps)$, but we can obtain the stretch $(2k-1)(1+\eps)$ by scaling $\eps \leftarrow \frac{\eps}{10g+1}$.  Note that \Cref{lm:LSConstruction} does not provide any bound on the number of edges of $H^{\sigma}_i$.

To bound the sparsity of $H^{\sigma}$ in our construction, we distinguish between \emph{isolated clusters} and \emph{non-isolated clusters}. A cluster $C\in \mathcal{C}_i$ is non-isolated if it   contains at least one endpoint of an edge in $E(H^{\sigma}_i)$, and otherwise, is isolated. By examining the construction of Le and Solomon carefully, we have that:

\begin{lemma}[Le and Solomon~\cite{LS21}]\label{lm:HiSizeYi}Let $\mathcal{Y}_i\subseteq \mathcal{C}_i$ be the set of all non-isolated clusters. Then $|E(H_i^{\sigma})| = O_{\eps}(n^{1/k})|\mathcal{Y}_i|$.
\end{lemma}

By property \hyperlink{P2L}{(2)}, the number of clusters is geometrically decreasing when $\eps$ is sufficiently smaller than $1$, and hence, the total number of clusters at all levels is $O(|\mathcal{C}|_1)$. This implies that:
\begin{equation}\label{eq:HsigmaLooseBound}
	|E(H^{\sigma})| = \sum_{i\geq 1}|E(H^{\sigma}_i)| = \sum_{i\geq 1} O_{\eps}(n^{1/k})|\mathcal{C}_i| =  O_{\eps}(n^{1/k}) |\mathcal{C}_1| =  O_{\eps}(n^{1/k}) |\tilde{V}|
\end{equation}

However, $|\tilde{V}|$ could be up to $\Omega(m)$ as it contains virtual vertices (\Cref{obs:virtualNum}). Thus, \Cref{eq:HsigmaLooseBound} does not provide any meaningful bound on the number of edges of $H^{\sigma}$.

We now describe our idea to modify LS algorithm and to bound the number of edges in $H^{\sigma}_i$. For each cluster $C\in \mathcal{C}_i$, we introduce two new types of clusters: \emph{non-virtual clusters}, denoted by $\mathcal{N}_i$, and \emph{virtual clusters}, denoted by $\mathcal{M}_i$.  A cluster $C\in \mathcal{C}_i$ is virtual if $C$ only contains virtual vertices, i.e., $C\subseteq \tilde{V}\setminus V$; otherwise $C$ is non-virtual. Since a non-isolated cluster contains at least one non-virtual vertex, which is the endpoint of an edge in $E(H^\sigma_i)$,  we have:
\begin{observation}\label{obs:} $\mathcal{Y}_i \subseteq \mathcal{N}_i$.
\end{observation}

Following the same idea of the construction in \Cref{sec:sparseRAM}, our goal is to construct a set of cluster $\mathcal{C}_{i+1}$ such that (in addition to properties in \Cref{lm:LSConstruction}) $|\mathcal{N}_{i}| - |\mathcal{N}_{i+1}| = \Omega(|\mathcal{Y}_i|)$. For notational convenience, we  define $\mathcal{N}_{i_{\max}+1} = \emptyset$.  By the same argument in  \Cref{sec:sparseRAM} and using \Cref{lm:HiSizeYi}, we could show that $|E(H^{\sigma})| = O_{\eps}(n^{1/k})|\mathcal{N}_1|$. Recall by the definition of non-virtual clusters that $|\mathcal{N}_1|\leq n$. It follows that  $|E(H^{\sigma})| = O_{\eps}(n^{1+1/k})$, which implies the desired sparsity bound. These ideas are formalized in the following lemma, whose proof is provided in \Cref{subsec:clustering}.

 \begin{lemma} \label{lm:NewConstruction} For each level $i\geq 1$, there is an algorithm that can compute a subgraph $H^{\sigma}_i$ \emph{induced by a subset of $ E_i^{\sigma}$}, as well as the set of level-$(i+1)$ clusters $\mathcal{C}_{i+1}$ satisfying properties \hyperlink{P1L}{(1)}-\hyperlink{P3L}{(3)}  given a set of clusters $\mathcal{C}_i$ at level $i$,  such that:  
 	\begin{enumerate}[noitemsep]
 		\item[(1)] $w(H^{\sigma}_i) =  O_{\eps}(n^{1/k}) \Delta_{i+1} + a_i$ for some $a_i\geq 0$ such that $\sum_{1\leq i\leq i_{\max}} a_i = O_{\eps}(n^{1/k})w(\mst)$.
 		\item[(2)] for every $(u,v)\in E^{\sigma}_i$, $d_{H^{\sigma}_{\leq i}}(u,v)\leq (2k-1)(1+(10g+1)\epsilon)w(u,v)$.
 		\item[(3)] $|E(H^\sigma_{i})| = O_{\eps}(n^{1/k})|\mathcal{Y}_i|$. 
 		\item[(4)] $|\mathcal{N}_{i}| - |\mathcal{N}_{i+1}| \geq |\mathcal{Y}_i|/2$. 
 	\end{enumerate}
 	Furthermore, the total running time of the construction of all levels is $O_{\eps}(m\alpha(m,n))$ in the pointer-machine model.
 \end{lemma}

In the next section, we prove \Cref{thm:2}, assuming that \Cref{lm:NewConstruction} holds.

\subsection{Proof of \Cref{thm:2}}\label{subsec:proofThm2}

Recall that $H = H_{\light} \cup H_{\heavy}\cup \mst$, where $H_{\light}$ is a $(2k-1)(1+\eps)$-spanner of $E_{\light}$. By \Cref{thm:1},  $H_{\light}$ can be constructed in  $O(m\alpha(m,n)\poly(1/\eps) + \srt(m))$ in the pointer-machine model, and in $O(m\poly(1/\eps))$ time  in the \wrdram model.  $\mst$ can be constructed in $O(m\alpha(m,n))$ by the pointer-machine model by Chazelle's algorithm~\cite{Chazelle00}. By \Cref{lm:NewConstruction}, the running time to construct $H^{\sigma}$ is $O(m\alpha(m,n)\poly(1/\eps))$, which implies the running time to construct $H$ is $O(m\alpha(m,n)\poly(1/\eps))\mu_{\eps} = O(m\alpha(m,n)\poly(1/\eps))$. Thus, the running time to construct $H$ is $O(m\alpha(m,n)\poly(1/\eps) + \srt(m))$ in the pointer-machine model and is $O(m\alpha(m,n)\poly(1/\eps))$  in the \wrdram model.

We now focus on bounding the sparsity and lightness of $H$. By Item (1) in \Cref{lm:NewConstruction}, we have that:
\begin{equation}\label{eq:HsigmaPreciseBound}
	\begin{split}
			w(H^{\sigma}) &= \sum_{i=0}^{i_{\max}}w(H^{\sigma}_i) =   \sum_{i=0}^{i_{\max}} O_{\eps}(n^{1/k})\Delta_{i+1} + a_i \\
			&= O_{\eps}(n^{1/k}) (\Phi_1) + \sum_{i=0}^{i_{\max}} a_i = O_{\eps}(n^{1/k})w(\mst),
	\end{split}
\end{equation}
by \Cref{eq:level1TotalPotential} and Item (2) of \Cref{lm:NewConstruction}. Furthermore, by Item (4) of \Cref{lm:NewConstruction}, it follows that:
\begin{equation}\label{eq:HsigmaSize}
	\begin{split}
			|E(H^{\sigma})|  &= \sum_{i=0}^{i_{\max}} |E(H^{\sigma}_i)| =  \sum_{i=0}^{i_{\max}} O_{\eps}(n^{1/k})|\mathcal{Y}_i| \quad \mbox{(by Item (3) of \Cref{lm:NewConstruction})} \\
			&=   \sum_{i=0}^{i_{\max}} O_{\eps}(n^{1/k}) (|\mathcal{N}_{i}| - |\mathcal{N}_{i+1}|)  \quad \mbox{(by Item (4) of \Cref{lm:NewConstruction})}\\
			&= O_{\eps}(n^{1/k})  |\mathcal{N}_1| = O_{\eps}(n^{1+1/k}).
	\end{split}
\end{equation}
It follows that $w(H_{\heavy}) = \sum_{\sigma \in [1,\mu_\eps]}w(H^{\sigma}) = O_{\eps}(n^{1/k})w(\mst)$ and $|E(H_{\heavy})| = \sum_{\sigma \in [1,\mu_\eps]}|E(H^{\sigma})| =  O_{\eps}(n^{1+1/k})$ since $\mu_\eps = O(\log(1/\eps)1/\eps)$.

Observe that $w(H_{\light})\leq w(E_{\light}) \leq w(\mst)/\eps$ by \Cref{eq:weightElight}. Furthermore, $|E(H_{\light})| = O_{\eps}(n^{1+1/k})$ by \Cref{thm:2}. We then conclude that:
\begin{equation*}
	\begin{split}
		w(H) &\leq w(H_{\light}) + w(H_{\heavy}) = O_{\eps}(n^{1/k})w(\mst)\\
		|E(H)|&= |E(H_{\light})| + |E(H_{\heavy})| = O_{\eps}(n^{1+1/k}).
	\end{split}
\end{equation*}
That is, the sparsity and lightness of $H$ are both $O_{\eps}(n^{1/k})$.

We now bound the stretch of $H$. Let $(u,v)$ be any edge in $G$. If $(u,v)$ is in $E_{\light}$, then the stretch of $(u,v)$ is $(2k-1)(1+\eps)$ in $H_{\light}$. If $(u,v)\in \mst$, then the stretch is $1$ since $H$ contains $\mst$ as a subgraph. Otherwise, $(u,v)\in E_{\heavy}$, and this means there exist $\sigma \in [1,\mu_\eps]$ and  $i\in [1,i_{\max}]$ such that $(u,v)\in E^{\sigma}_i$. By Item (1) in \Cref{lm:NewConstruction}, the stretch of $(u,v)$ in $H^{\sigma}_{\leq i}$, and hence in $H_{\heavy}$, is $(2k-1)(1+(10g+1)\epsilon)$. In summary, the stretch in $H$ of any edge $(u,v)\in E(G)$ is at most $(2k-1)(1+(10g+1)\epsilon)$. By scaling $\eps \leftarrow \eps/(10g+1)$, we obtain a spanner of stretch $(2k-1)(1+\epsilon)$, and with the same lightness and sparsity bounds.  \qed

\subsection{Construction of $H^{\sigma}_i$ and $\mathcal{C}_{i+1}$}\label{subsec:clustering}
 
In this section, we construct $H^{\sigma}_i$ and $\mathcal{C}_{i+1}$ with properties claimed in \Cref{lm:NewConstruction}. Without loss of generality, we assume that $\eps$ is sufficiently small, and in particular, $\eps$ is smaller than $1/(c\cdot g)$ for any constant $c$. We now introduce new notation used in this section.

\paragraph{Notation.~} We consider graphs with weights on both \emph{edges and vertices} in this section.  We define the \emph{augmented weight} of a path to be the total weight of all edges and vertices along the path. The \emph{augmented distance} between two vertices in $G$ is defined as the minimum augmented weight of a path between them in $G$. The augmented diameter of $G$ is denoted by  $\adm(G)$, which is the maximum pairwise augmented distance in $G$.

\paragraph{Cluster graphs.~}  The construction of $\mathcal{C}_{i+1}$ is done via a \emph{cluster graph} $\mg_i(\mv_i, \me_i',\omega)$ that has weights on both edges and nodes (we use nodes to refer to vertices of $\mg_i$). Each node $\varphi_C \in \mv_i$ corresponds to a level-$i$ cluster $C$ and has weight:
\begin{equation}\label{eq:nodeWeight}
	\omega(\varphi_C) = \Phi(C)
\end{equation}
That is, the weight of each node is the \emph{potential value} of its corresponding cluster.   The edge set $\me'_i$ is the union of two edge sets $\me_i\cup \msttilde_i$:
\begin{itemize}[noitemsep]
	\item Each edge $\mbe = (\varphi_{C_u}, \varphi_{C_v})\in \me_i$ corresponds to an edge $(u,v)\in E^{\sigma}_i\cup E(\msttilde)$ where $C_u$ and $C_v$ are level-$i$ clusters containing $u$ and $v$, respectively. Furthermore, $\omega(\mbe) = w(u,v)$.
	\item $\me_i$ corresponds to a subset of edges of $E^{\sigma}_i$ and $\msttilde_i$ corresponds to a subset of edges of $\msttilde$, the subdivided $\mst$. $\msttilde_i$ induces a minimum spanning tree of $\mg_i$, and we abuse notation by denoting $\msttilde_i$ the spanning tree of $\mg_i$ by edges in $\msttilde_i$.
\end{itemize}  
 We refer readers to Lemma 3.16 in \cite{LS21} for the construction of $\mg_i$. At a high level, the construction removes edges that are self-loops, parallel edges, and those that have stretch at most $(2k-1)(1+6g\eps)$ in $\msttilde_i$ as these edges already have a good stretch. 

\begin{lemma}[Lemma 3.16 and Lemma 3.22~\cite{LS21}]\label{lm:GiConstr} $\mg_i$ can be constructed in $O((|\mv_i| + |E^{\sigma}_i|)\alpha(m,n))$ time. Furthermore, if the subset of edges of $E^{\sigma}_i$ corresponding to  $\me_i$ has stretch $(2k-1)(1+s\eps)$ in $H^{\sigma}_{\leq i}$ for some constant $s$ that only depends on $g$, then every edge in $E^{\sigma}_i$ has stretch $(2k-1)(1 + \max\{s+4g, 10g\}\eps)$ in  $H^{\sigma}_{\leq i}$ when $\eps \leq\frac{1}{\max\{s+4g, 10g\}}$. 
\end{lemma}

\Cref{lm:GiConstr} implies that it suffices for the construction to focus on constructing a spanner for the subset of edges of $E^{\sigma}_i$ that correspond to edges in $\me_i$.

\paragraph{Level-$(i+1)$ clusters.~} Instead of constructing level-$(i+1)$ directly from level-$i$ clusters, we construct a collection $\mathbb{X}$ of vertex-disjoint subgraphs of $\mg_i$. Each subgraph $\mx\in \mathbb{X}$ has the vertex set denoted by $\mv(\mx)$ and the edge set denoted by $\me(\mx)$, and is mapped to a level-$(i+1)$ cluster, denoted by $C_{\mx}$, as follows:
\begin{equation}\label{eq:XtoCluster}
	C_{\mx} = \cup_{\varphi_C\in \mv(\mx)}C
\end{equation}
That is, $C_{\mx}$ is the union of all level-$i$ clusters corresponding to the nodes of $\mx$. Note that $\mx$ has weights on both edges and nodes. We then define the potential value of $C_{\mx}$ as follows.
\begin{equation}\label{eq:CXPotential}
	\Phi(C_{\mx}) = \adm(\mx)	
\end{equation}
That is, the potential value  of $C_{\mx}$ is the augmented diameter of the corresponding subgraph. Recall that the potential value will then be the weight of the node corresponding to $C_{\mx}$ in the cluster graph $\mg_{i+1}$ in the construction of  the next level, see \Cref{eq:nodeWeight}. Furthermore, inductively, we can show that, if $\omega(\varphi_C)$ is an upper bound on $\dm(H^{\sigma}_{\leq i-1}[C])$, then $\adm(\mx)$ is an upper bound on $\dm(H^{\sigma}_{\leq i}[C_{\mx}])$. As a result, guaranteeing properties \hyperlink{P1L}{(1)}-\hyperlink{P3L}{(3)} for level-$(i+1)$ clusters can be translated into guaranteeing the following properties for subgraphs in $\mathbb{X}$:

\begin{itemize}[noitemsep]
	\item \textbf{(P1').~} \hypertarget{P1'L}{}  $\{\mv(\mx)\}_{\mx \in \mathbb{X}}$ is a partition of $\mv_i$.
	\item \textbf{(P2').~} \hypertarget{P2'L}{} $|\mv(\mx)| = \Omega(\frac{1}{\eps})$.
	\item \textbf{(P3').~} \hypertarget{P3'L}{} $L_i \leq \adm(\mx) \leq gL_{i}$.
\end{itemize}

\begin{lemma}[Lemma 3.14~\cite{LS21}]\label{lm:XInvariants} Let $\mx$ be any subgraph in $\mathbb{X}$ satisfying properties \hyperlink{P1'L}{(P1')}-\hyperlink{P3'L}{(P3')}. Suppose that every edge $(\varphi_{C_u},\varphi_{C_v})\in \me(\mx)$ corresponds to an edge $(u,v)$ that is added to $H^{\sigma}_{i}$. Then, $C_{\mx}$  satisfies all properties \hyperlink{P1L}{(1)}-\hyperlink{P3L}{(3)}.
\end{lemma}

We remark that \Cref{lm:XInvariants} is based on the assumption that  $(u,v)$ is added to $H^{\sigma}_{i}$, which we have not constructed yet. 

\paragraph{Constructing level$(i+1)$ clusters.~} \Cref{lm:XInvariants} translates the construction of clusters in $\mc_{i+1}$ to the construction of the set of subgraphs $\mathbb{X}$ satisfying \hyperlink{P1'L}{(P1')}-\hyperlink{P3'L}{(P3')}. The main difficulty is not only to satisfy these properties; but also to guarantee that the weight of $H^{\sigma}_i$ is bounded by the potential change $\Delta_{i+1}$ (and a small additive term) as claimed in Item (1) of \Cref{lm:NewConstruction}. Recall by \Cref{eq:nodeWeight} and \Cref{eq:CXPotential} that: 
\begin{equation}\label{eq:Phii}
	\begin{split}
			\Phi_i &= \sum_{C\in \mathcal{C}_{i}} \Phi(C) = \sum_{\varphi_C \in \mv_i}\omega(\varphi_C)\\
			\Phi_{i+1} &= \sum_{C_{\mx}\in \mathcal{C}_{i+1}} \Phi(C_{\mx}) = \sum_{\mx\in \mathbb{X}}\adm(\mx)
	\end{split}
\end{equation}
Thus, if we define the \emph{local potential change} of $\mx$ as follows:
\begin{equation}\label{eq:localChangeX}
	\Delta_{i+1}(\mx) = \sum_{\varphi_C\in \mv(\mx)}\omega(\varphi_{\mx})-\adm(\mx), 
\end{equation} 

\noindent then it follows that:

\begin{claim}[Claim 3.15~\cite{LS21}]\label{clm:PotentialDecompos} $\Delta_{i+1} = \sum_{\mx\in \mathbb{X}}\Delta_{i+1}(\mx)$.
\end{claim}

That is, the potential change $\Delta_{i+1}$  can be decomposed into local potential changes of subgraphs in $\mathbb{X}$. This meanss we could bound the weight of $H_i^{\sigma}$ \emph{locally} via  bounding the total weight of edges incident to nodes in $\mx$ by the local potential change of $\mx$. 

\paragraph{Partitioning $\mv_i$ and $\mathbb{X}$.~}
We say that a partition $\mathbb{V} = \{\mv_i^{\high}, \mv_i^{\lowp},\mv_i^{\lowm}\}$ of $\mv_i$ is a \emph{degree-specific} partition if  every node $\varphi_C \in \mv_i^{\high}$ is incident to  $\Omega(1/\eps)$ edges in $\me_i$ and every node $\varphi_C \in \mv_i^{\lowp} \cup \mv_i^{\lowm}$ is incident to $O(1/\eps)$ edges in $\me_i$.  That is, $\mv^{\high}_i$ is the set of high-degree nodes of $\mv_i$ and $\mv_i^{\lowp} \cup \mv_i^{\lowm}$ is the set of low-degree nodes of $\mv_i$. The difference between $\mv_i^{\lowp}$ and $\mv_i^{\lowm}$ will be made clear later.
	
We say that a partition  $\{\mathbb{X}^{\high}, \mathbb{X}^{\lowp}, \mathbb{X}^{\lowm}\}$ of a  collection $\mathbb{X}$ of subgraphs of $\mg_i$ \emph{conforms} with a degree-specific partition $\mathbb{V}$ if
\begin{itemize}[noitemsep]
	\item[(i)] Every subgraph $\mx \in \mathbb{X}^{\lowm}$ has $\mv(\mx)\subseteq \mv_i^{\lowm}$, and $\cup_{\mx \in \mathbb{X}^{\lowm}}\mv(\mx) = \mv_i^{\lowm}$.
	\item[(ii)] For every node $\varphi_C \in \mv_i^{\high}$, there exists a subgraph $\mx \in \mathbb{X}^{\high}$ such that $\varphi_C \in \mv(\mx)$, and that every subgraph $\mx \in \mathbb{X}^{\high}$ contains at least one node in $\mv_i^{\high}$.
\end{itemize}

Observe that property (ii) implies that $\mv(\mx)\subseteq \mv_i^{\lowp}$ for every $\mx \in \mathbb{X}^{\lowp}$. Also, it is possible that a subgraph $\mx \in \mathbb{X}^{\high}$ contains a node in $\mv_i^{\lowp}$. 

The construction of $\mathbb{X}$ in~\cite{LS21} is described by the following lemma. 

\begin{restatable}[Lemma 3.17~\cite{LS21}]{lemma}{Clustering}
	\label{lm:Clustering}  Given $\mg_i$, we can construct in time $O((|\mv_i| + |\me_i|)\eps^{-1})$ (i) a degree-specific partition $\mathbb{V} = \{\mv_i^{\high}, \mv_i^{\lowp},\mv_i^{\lowm}\}$  of $\mv_i$ and (ii) a collection $\mathbb{X}$ of subgraphs of $\mg_i$ along with a partition  $\{\mathbb{X}^{\high}, \mathbb{X}^{\lowp}, \mathbb{X}^{\lowm}\}$ conforming $\mathbb{V}$ such that:
	\begin{enumerate}
		\item[(1)] Let $\Delta_{i+1}^+(\mx) = \Delta(\mx) + \sum_{\mbe \in \msttilde_i\cap \me(\mx)}w(\mbe)$.  Then, $\Delta_{i+1}^+(\mx) \geq 0$ for every $\mx \in \mathbb{X}$, and
		\begin{equation}\label{eq:averagePotential}
			\sum_{\mx \in\mathbb{X}^{\high}\cup \mathbb{X}^{\lowp}} \Delta_{i+1}^+(\mx) = \sum_{\mx \in \mathbb{X}^{\high}\cup \mathbb{X}^{\lowp}} \Omega(|\mv(\mx)|\eps^2 L_i). 
		\end{equation}
		\item[(2)] There is no edge in $\me_i$ between a node in $\mv^{\high}_i$ and a node in $\mv^{\lowm}_i$. Furthermore, if there exists an edge   $(\varphi_{C_u},\varphi_{C_v}) \in \me_i$ such that both $\varphi_{C_u}$ and $\varphi_{C_v}$ are in 
		$\mv_i^{\lowm}$, then  $\mv_i^{\lowm} = \mv_i$ and $|\me_i| = O(\frac{1}{\epsilon^2})$; this case is called the \emph{degenerate case}.
		\item[(3)] For every subgraph $\mx \in \mathbb{X}$, $\mx$ satisfies the three properties (\hyperlink{P1'L}{P1'})-(\hyperlink{P3'L}{P3'}) with constant $g=31$,  and $|\me(\mx)\cap \me_i| = O(|\ma_{\mx}|)$ where $\ma_{\mx}$ is the set of nodes in $\mx$ incident to an edge in $\me(\mx)\cap \me_i$.
	\end{enumerate}	
\end{restatable}

We call $\Delta_{i+1}(\mx)$ the \emph{corrected potential change} of $\mx$. We remark that $\Delta_{i+1}(\mx)$ could be negative but $\Delta_{i+1}(\mx)$  is always positive by Item (1) of \Cref{lm:Clustering}. Furthermore, Item (1) in \Cref{lm:Clustering} only tells us about the corrected potential changes of subgraphs in $\mathbb{X}^{\high}\cup \mathbb{X}^{\lowp}$; there is no guarantee on the corrected potential changes of subgraphs in $\mathbb{X}^{\lowm}$ other than non-negativity, and as a result, we could not bound the total weight of edges incident to a subgraph $\mx \in \mathbb{X}^{\lowm}$ by the local potential change of $\mx$. However, Item (2) means that subgraphs in  $\mathbb{X}^{\lowm}$ do not need to ``pay for'' their incident edges (by their corrected potential changes)---these edges can be paid for by subgraphs in $\mathbb{X}^{\high}\cup \mathbb{X}^{\lowp}$---unless the degenerate case happens, which only incurs a small weight (of $O(1/\eps^2)$ edges). Furthermore,  subgraphs in $\mathbb{X}^{\lowm}$ do not contain any edge in $\me_i$ by Item (2) of \Cref{lm:Clustering} unless the degenerate case happens.

\begin{observation}[Observation 3.20 in \cite{LS21}]\label{obs:LowmStructure} If the degenerate case does not happen, for every edge $(\varphi_{1},\varphi_{2})$ with one endpoint in $\mv_i^{\lowm}$, the other endpoint must be in $\mv_i^{\lowp}$, and hence, $\me(\mx)\cap \me_i = \emptyset$ if $\mx \in \mathbb{X}^{\lowm}$.
\end{observation}

We remark that Item (3) in \Cref{lm:Clustering} is slightly different from the corresponding item in Lemma 3.17~\cite{LS21}, which is Item (5), in that  $|\me(\mx)\cap \me_i|$ is bounded by $O(|\mv(\mx)|)$. Here we need a slightly stronger bound, and Item (3) can be seen directly from the construction of ~\cite{LS21}. For completeness, we will show this item in the construction in \Cref{subsubsec:clustering}.

While the construction in \Cref{lm:Clustering} provides a mean to construct $H_i^{\sigma}$ and bounding its weight by (corrected) potential changes via Item (1), it does not give us a sufficient reduction in the number of non-virtual clusters  as claimed by Items (3) and (4) in \Cref{lm:NewConstruction}. The reduction in the number of non-virtual clusters was used to bound  the total number of edges of $H^{\sigma}$ in \Cref{subsec:proofThm2}. Our  main contribution is a modification of the construction by Le and Solomon~\cite{LS21} using the cycle property of $\mst$ to achieve the reduction in the number of non-virtual clusters. 

We call a node $\varphi_C$ \emph{virtual} if it corresponds to a virtual cluster $C$; otherwise, we call $\varphi_C$ \emph{non-virtual}.  We say that $\varphi_C$ is isolated if $C$ is isolated, and otherwise, is non-isolated. By definition, a non-isolated node is a non-virtual node.

We abuse notation by denoting $\mathcal{N}_i$ and $\mathcal{M}_i$ the sets of non-virtual nodes and virtual nodes of $\mv_i$, respectively. We denote by $\my_i$ the set of non-isolated nodes in $\mv_i$. We will show later that $\my_i$ is exactly the set of nodes defined in \Cref{lm:HiSizeYi}. That is, every node in $\my_i$ corresponds to a level-$i$ cluster that contains at least one endpoint of an edge in $H^{\sigma}_i$.  

We say that a subgraph $\mx \in \mathbb{X}$ is non-virtual if it contains at least one non-virtual node, and otherwise, is virtual.  A non-virtual subgraph corresponds to a non-virtual level-$(i+1)$ cluster. Our main contribution is the construction of $\mathbb{X}$ described by the following lemma.

\begin{lemma}\label{lm:additional-prop}We can construct in  $O((|\mv_i| + |\me_i|)\eps^{-1})$ a  degree-specific partition $\mathcal{V}$ of $\mv_i$ and a collection $\mathbb{X}$ of subgraphs of $\mg_i$ that satisfy all properties in \Cref{lm:Clustering} with $g = 42$. Furthermore, if we denote by $\mathcal{N}_{i+1}$ the set of non-virtual subgraphs in $\mathbb{X}$, then $|\mathcal{N}_i| - |\mathcal{N}_{i+1}| \geq |\my_i|/2$.
\end{lemma}

In the following section, we prove \Cref{lm:NewConstruction} assuming that \Cref{lm:additional-prop} holds. The proof of \Cref{lm:additional-prop} is deferred to \Cref{subsubsec:clustering}.

\subsubsection{Proof of \Cref{lm:NewConstruction}}

We use the same algorithm in \cite{LS21} to construct $H_i^{\sigma}$. The algorithm has three steps. Initially $H_i^{\sigma}$ has no edge.  
\begin{itemize}[noitemsep]
	\item \textbf{Step 1.~} For every subgraph $\mx \in \mathbb{X}$, we add to $H_i^{\sigma}$ every edge in $E^{\sigma}_i$ that corresponds to an edge in $\me(\mx)\cap \me_i$. The purpose of this step is to guarantee the assumption of \Cref{lm:XInvariants}.
	\item \textbf{Step 2.~}  Wee use Halperin-Zwick algorithm (\Cref{thm:HalperinZwick}) to construct  a $(2k-1)$-spanner for edges between $\mv_i^{\high}$ only. Specifically, we create an \emph{unweighted} graph $\mathcal{K}_i$  that has $\mv_i^{\high}$ as the vertex set and the subset of edges of $\me_i$ between $\mv_i^{\high}$ as the edge set. Then, we run Halperin-Zwick algorithm~\cite{HZ96} on $\mathcal{K}_i$ to obtain an edge set $\me^{\prune}_i$. We then add every edge in $E^{\sigma}_i$ corresponding to an edge in $\me^{\prune}_i$ to $H^{\sigma}_i$.
	\item \textbf{Step 3.~} We add to $H^{\sigma}_i$  every edge that corresponding to an edge of $\me_i$ incident to a node in $\mv_i^{\lowp}\cup \mv_i^{\lowm}$.
\end{itemize}
 
 Le and Solomon (Lemma 3.22 and Lemma 4.5 in ~\cite{LS21}) showed that $w(H^{\sigma}_i) = O_{\eps}(n^{1/k})\Delta_{i+1} + a_i$ for $a_i$ satisfying \Cref{lm:NewConstruction}, and that the stretch of every edge $(u,v) \in E^{\sigma}_i$ is at most $(2k-1)(1+(10g+1)\epsilon)$ in $H^{\sigma}_{i}$. Since their proof only uses properties stated in \Cref{lm:Clustering}, and that our construction in \Cref{lm:additional-prop} also satisfies \Cref{lm:Clustering}, Items (1) and (2) in \Cref{lm:NewConstruction} hold in our construction as well. We remark that the additive term $a_i$ is used to handle the degenerate case in Item (2) of \Cref{lm:Clustering}, since in that case, $\Delta_{i+1} \leq  0$. 
 
 We now focus on proving Items (3) and (4) of \Cref{lm:NewConstruction}. First, we observe that 
  for every node $\varphi_C$ that is incident to an edge  $\mbe \in \me_i$, the corresponding edge of $\mbe$ in $E^{\sigma}_i$ is added to $H^{\sigma}_i$, unless $\varphi_C\in \mv_i^{\high}$ and  Halperin-Zwick algorithm does not pick $\mbe$ to $\me^{\prune}_i$.  In this exceptional case, another edge incident to $\varphi_C$ must be picked to  $\me^{\prune}_i$; otherwise, $\varphi_C$ is  not connected to any node in the graph induced by $\me^{\prune}_i$, contradicting that the output is a spanner. It follows that $\my_i$ corresponds to level-$i$ clusters that have at least one incident edge in $H^{\sigma}_i$. Thus, Item (4) of \Cref{lm:NewConstruction} follows from \Cref{lm:additional-prop}. 
 
 By Item (3) in \Cref{lm:Clustering}, the total number of edges added in Step 1 is $O(1)\sum_{\mx \in \mathbb{X}}\ma_{\mx} = O(1)|\my_i|$. The number of edges added in Step 2 is  $|\me_i^{\prune}|= O(|\mv_i^{\high}|^{1+1/k}) = O(n^{1/k})|\mv_i^{\high}| = O(n^{1/k})|\my_i|$ since $\mv_i^{\high}\subseteq \my_i$ by the definition of non-isolated nodes. In Step 3, for each node $\varphi_C \in \mv_i^{\lowp}\cup \mv_i^{\lowm}$, we add at most $O(1/\eps)$ incident edges to $H^{\sigma}_i$ since nodes in $\mv_i^{\lowp}\cup \mv_i^{\lowm}$ have degree $O(1/\eps)$. Thus, the total number of edges added in Step 3 is $O(1/\eps)|\my_i|$. Item (3) of \Cref{lm:NewConstruction} now follows.   
 
 For the running time, we first note that constructing $\mg_i$ takes $O((|\mv_i| + |\me_i|)\alpha(m,n))$ time by \Cref{lm:GiConstr}. The set of subgraphs $\mathbb{X}$ is constructed in $O_{\eps}(|\mv_i| + |\me_i|)$ time by \Cref{lm:additional-prop}. In the construction of $H^{i}_{\sigma}$, Steps 1 and 3 take $O(|\mv_i| + |\me_i|)$ by a straightforward implementation. Step 2 takes $O(|\mv_i| + |\me_i|)$ time by \Cref{thm:HalperinZwick}. Thus, the total running time of the construction at level $i$ is $O((|\mv_i| + |\me_i|)\alpha(m,n))$. It follows that the total running time over all levels is:
 \begin{equation*}
 	\begin{split}
 		 	\sum_{i\geq 1} O_{\eps}\left((|\mv_i| + |\me_i|)\alpha(m,n)\right) & =	 O_{\eps}\left(( \sum_{i\geq 1}|\mv_i| + m )\alpha(m,n)\right)\\
 		 	& =  O_{\eps}\left((|\tilde{V}| + m )\alpha(m,n)\right)\qquad \mbox{(by \hyperlink{P2L}{property (P2)})} \\ 
 		 	& =   O_{\eps}\left(m\alpha(m,n)\right)  \qquad \mbox{(by \Cref{obs:virtualNum})}
 	\end{split}
 \end{equation*}
 \Cref{lm:NewConstruction} now follows. \qed

\subsubsection{The construction of $\mathbb{X}$}\label{subsubsec:clustering}

Recall that $\msttilde$ is the tree obtained by subdividing $\mst$ edges by virtual vertices. For each edge $e \in \mst$, we denote by $P_e$ the corresponding path of $\mst$ subdivided from $e$. We call $P_e$ the \emph{subdivided path} of $e$.  Since each virtual cluster $C\in \mathcal{C}_i$ only contains virtual vertices, $C$ induces a subpath of the subdivided path $P_e$ of some edge $e \in \mst$. We call $P_e$ the \emph{parent path} of $C$, and $e$ the \emph{parent edge} of $C$; see Figure~\ref{fig:cycleProp}(a). We also refer to $P_e$ as the parent path and to $e$ as the parent edge of the virtual node $\varphi_{C}$ corresponding to $C$.

\begin{figure}[!h]
	\begin{center}
		\includegraphics[width=1.0\textwidth]{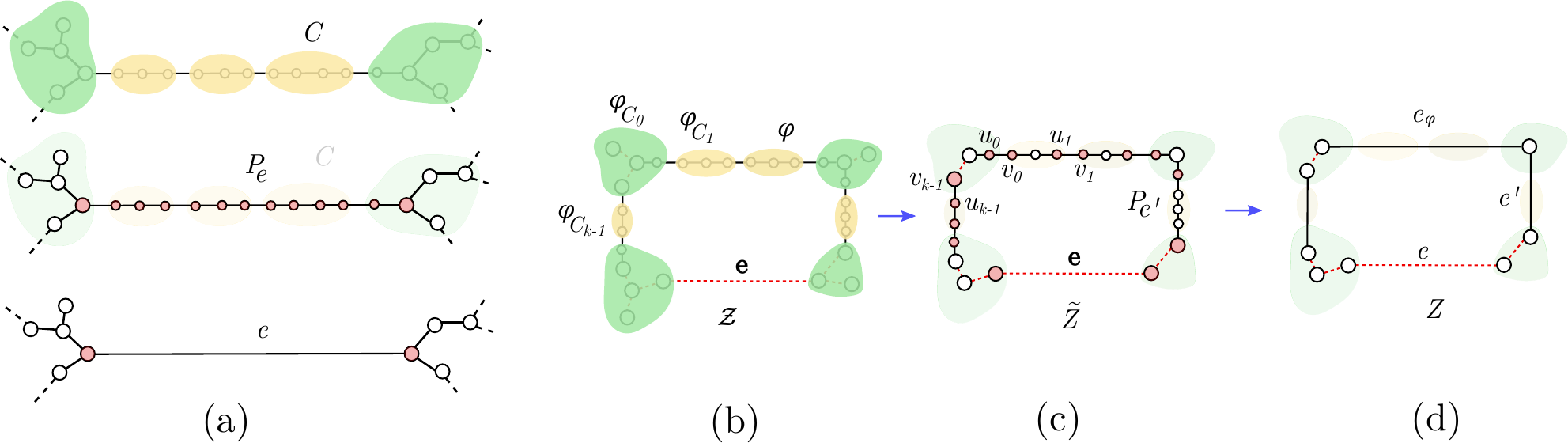}
	\end{center}
	\caption{Virtual clusters are yellow shaded and non-virtual clusters are green shaded. (a) A virtual cluster $C$, its parent path $P_e$, and the edge $e$ that corresponds to $P_e$. (b) The fundamental cycle $\mathcal{Z}$ of $\msttilde_{i}$ formed by an edge $\mbe$. (c) The corresponding cycle $\tilde{Z}$ in $\tilde{G}_{\heavy}$ corresponding to $\mathcal{Z}$. Shaded nodes are $u_0,v_0,u_1,v_1,\ldots,u_{k-1},v_{k-1}$. (d) The cycle $Z$  of $G_{\heavy}$ corresponding to $\tilde{Z}$ obtained by replacing each subdivided path $P_{e'}$ with the corresponding edge $e'$ in $\mst$. Solid (black) edges are $\mst$ edges, and red (dashed) edges are non-$\mst$ edges.}
	\label{fig:cycleProp}
\end{figure}

Our goal is to construct $\mathbb{X}$ satisfying all properties in \Cref{lm:Clustering}, and such that there is a significant reduction in the number of non-isolated clusters as claimed in \Cref{lm:NewConstruction}. To guarantee this additional constraint, we rely on a specific structure of  $\mg_i$ described in the following lemma, which is an analogous version of the cycle property of the minimum spanning tree.

\begin{lemma}\label{lm:cycle-Prop}  Let $\mbe = (\varphi_1,\varphi_2)$ be any edge in $\me_i$, and $\mz$ the fundamental cycle of $\msttilde_i$ formed by $\mbe$. For any virtual node $\varphi \in \mz$, $w(e_{\varphi}) \leq \omega(\mbe)$ where $e_{\varphi}$ is the parent edge of $\varphi$.  
\end{lemma}
\begin{proof}
	Recall that $\tilde{G}_{\heavy}$ is obtained from $G_{\heavy}$ by subdividing $\mst$ edges, and that $G_{\heavy} = (V, E_{\heavy}\cup E(\mst))$. Let $e$ be the edge in $G_{\heavy}$ corresponding to $\mbe$. We construct a cycle $\tilde{Z}$ of $\tilde{G}_{\heavy}$ from $\mathcal{Z}$ as follows.  Write $$\mathcal{Z} = (\varphi_{C_0},\mbe_0,\varphi_{C_1},\mbe_1,\ldots,\varphi_{C_{k-1}}, \mbe_{k-1}, \varphi_{C_0})$$ as an alternating sequence of nodes and edges that starts from and ends at the same node $\varphi_{C_0}$. (See \Cref{fig:cycleProp}(a) and (b) for an illustration.)  For notational convenience, we regard the last node $\varphi_{C_0}$ as $\varphi_{C_k}$ with the subscript modulo $k$.   Let $(u_i,v_i)$ be the edge in $\tilde{G}_{\heavy}$ corresponding to $\mbe_i$, and $Q_i$ be the shortest path from $v_{(i-1)\mod k}$  to $u_i$ in $H^{\sigma}_{\leq i}[C_i]$ for any $0\leq i \leq k-1$. Then $\tilde{Z} = (u_0,v_0)\circ Q_1\circ  (u_1,v_1) \circ Q_2 \circ\ldots \circ (u_{k-1},v_{k-1})\circ Q_{k-1}$  is a cycle of $\tilde{G}$; here $\circ$ is the path concatenation operator. Observe that $\tilde{Z}$ contains the parent path, say $P_e$,  of $\varphi$. Let $Z$ be the cycle of $G_{\heavy}$ obtained from $\tilde{Z}$ by replacing each subdivided path say $P_{e'}$ in $\tilde{Z}$ with the corresponding $\mst$ edge $e'$; see Figure~\ref{fig:cycleProp}(d).  Note that both $e$ and $e_{\varphi}$ belong to $Z$. 
		
	Observe by property \hyperlink{P3L}{(P3)} that $\dm(H^{\sigma}_{\leq i}[C_i]) \leq g\eps L_i < L_i/(1+\eps) \leq \omega(\mbe) = w(e)$ when $\eps \leq 1/(2g)$. Thus, the weight of any non-$\mst$ edge in $Z$ is at most $w(e)$. That is, any edge of weight larger than $w(e)$ in $Z$ must be an $\mst$ edge. If there exists such an edge, then the edge of maximum weight in $Z$ is an $\mst$ edge, contradicting the cycle property of $\mst$. Thus, $e$ is  an edge of maximum weight in $Z$, which gives $w(e_{\varphi}) \leq w(e) = \omega(\mbe)$ as claimed. \qed
\end{proof}

Note by definition that a non-isolated node is a non-virtual node.  We say that a subgraph $\mx$ is \emph{good} if it either contains no non-isolated node or if it contains one non-isolated node, it has at least two non-virtual nodes (one of which is the non-isolated node). If every subgraph in $\mathbb{X}$ is good, then we could show that the number of non-virtual  clusters is reduced by at least $|\my_i|/2$. In LS construction, which has five steps, only subgraphs formed in Steps 2 and 5 (more precisely, Step 5B) may not be good. For Step 5B,  only need to make a minor modification and argue that the resulting subgraph is good  using \Cref{lm:cycle-Prop}.  For Step 2, we need an entirely different construction. As a result, our construction also has five steps. Steps 1,3, 4 and 5A are the same as the LS construction, and are taken verbatim from~\cite{LS21} for completeness. Notation introduced in this section is summarized in the following table.

\renewcommand{\arraystretch}{1.3}
\begin{longtable}{| l | l|} 
	\hline
	\textbf{Notation} & \textbf{Meaning} \\ \hline
	$E_{\light}$ &$ \{e \in E(G) : w(e)\le w/\varepsilon\}$\\ \hline 
	$E_{\heavy}$ & $E \setminus E^{light}$ \\\hline
	$E^{\sigma} $ & $\bigcup_{i \in \mathbb{N}^{+}}E_{i}^{\sigma}$\\\hline
	$E_{i}^{\sigma} $ & $\{e \in E(G) : \frac{L_i}{1+\eps} < w(e) \le L_i\}$\\\hline
	$H^\sigma_i$ & A spanner constructed for edges in $E^{\sigma}_i$\\ \hline	
	$H^\sigma_{\leq i}$ & $H^\sigma_{\leq i} =\cup_{j\leq i}H^{\sigma}_i$\\ \hline
	$g$ & constant in \hyperlink{P3L}{property (P3)}, $g = 42$ \\\hline
	Non-virtual cluster& A cluster containing at least one non-virtual vertex\\\hline
	Non-virtual node& A node in $\mv_i$ corresponding to a non-virtual cluster\\\hline
	$\mathcal{N}_i$ & the set of non-virtual clusters (nodes) at level $i$\\\hline
	$\mathcal{M}_i$ & the set of virtual clusters (nodes) at level $i$\\\hline
	Non-isolated cluster& A cluster containing an endpoint of an edge added to $H^\sigma_i$\\ \hline 
	Non-isolated node& A node in $\mv_i$ corresponding to a non-isolated cluster\\\hline
	$\mathcal{Y}_i$ & the set of non-isolated clusters (nodes) at level $i$; $\my_i\subseteq \mathcal{N}_i$\\\hline
	$\mathcal{G}_i = (\mv_i, \msttilde_{i} \cup \mathcal{E}_i, \omega)$ & cluster graph \\\hline
	$\me_i$ & corresponds to a subset of edges of $E^{\sigma}_i$\\\hline
	$\mathbb{X}$ & a collection of subgraphs of $\mathcal{G}_i$\\\hline
	$\mx, \mv(\mx), \me(\mx)$ & a subgraph in $\mathbb{X}$, its vertex set, and its edge set\\\hline
	Good subgraph $\mx$& $\mx$ contains no non-isolated node or at least two non-virtual nodes  \\\hline
	$\Phi_i$ & $\sum_{c \in C_i}\Phi(c)$ \\\hline
	$\Delta_{i+1} $&$ \Phi_i - \Phi_{i+1}$\\\hline
	$\Delta_{i+1}(\mx)$ & $(\sum_{\phi_C\in \mx }\Phi(C) ) - \Phi(C_{\mx})$\\\hline
	$\Delta_{i+1}^+(\mx)$ & $\Delta_{i+1}(\mx) + \bigcup_{e \in \me(\mx)\cap \msttilde_i}w(e)$\\\hline
	$C_\mx$ & $\bigcup_{\phi_C \in \mx}C$ \\\hline
	$\{\mv^{\high}_i,\mv^{\lowp},\mv^{\lowm}_{i}\}$ & a degree-specific partition of $\mv_i$ \\\hline
	$\{\mathbb{X}^{\high}, \mathbb{X}^{\lowp}, \mathbb{X}^{\lowm}\}$ & A partition of $\mathbb{X}$ conforming a degree-specific partition.  \\\hline
	\caption{Notation introduced in  this section}
	\label{table:notation}
\end{longtable}
\renewcommand{\arraystretch}{1}

\begin{lemma}[Step 1, Lemma 5.1~\cite{LS21}]\label{lm:Clustering-Step1} Let $\mv^{\high}_i$ be the set of nodes incident to at least $2g/\eps$ edges in $\me_i$, and $\mv^{\high+}_i$ be the set of all nodes in $\mv^{\high}_i$ and their neighbors that are connected via edges in $\me_i$. We can construct in $O(|\mv_i| + |\me_i|)$ time a collection of node-disjoint subgraphs $\mathbb{X}_1$ of $\mg_i$ such that:
	\begin{enumerate}[noitemsep]
		\item[(1)] Each subgraph $\mx \in \mathbb{X}_1$ is a tree.
		\item[(2)] $\cup_{\mx \in \mathbb{X}_1}\mv(\mx) = \mv^{\high+}_i$.
		\item[(3)] $L_i \leq \adm(\mx) \leq 13L_i$, assuming that $\eps \leq 1/g$. 
		\item[(4)] $|\mv(\mx)|\geq \frac{2g}{\eps}$.
	\end{enumerate}
\end{lemma}

 Let $\Ftilde^{(2)}_i$ be the forest obtained from $\msttilde_{i}$ by removing every node in $\mv^{\high+}_i$ (defined in \Cref{lm:Clustering-Step1}). LS algorithm deals with branching nodes of $\Ftilde^{(2)}$ in Step 2. We say that a node in a tree $\Ttilde$ is \emph{$\Ttilde$-branching}  if it has degree at least $3$ in $\Ttilde$. A node in a forest $\Ftilde$ is $\Ftilde$-branching if it is $\Ttilde$-branching in some tree $\Ttilde$ of $\Ftilde$. We will omit the prefixes $\Ttilde$ and $\Ftilde$ in the branching notation whenever the tree and the forest are clear from the context.

 Similar to LS algorithm, our goal is to group all branching nodes of  $\Ftilde^{(2)}_i$ into subgraphs.  However, we need to guarantee that subgraphs formed in this step are good, which a priori,  are not guaranteed to be good in LS construction.

\begin{restatable}{lemma}{ClusteringStepTwo}
	\label{lm:Clustering-Step2}   We can construct in $O(|\mv_i|)$ time a collection $\mathbb{X}_2$ of subtrees of $\Ftilde^{(2)}_i$ and a subset of nodes $\mz$ of $\Ftilde^{(2)}_i$ such that, for every $\mx \in \mathbb{X}_2$:
	\begin{enumerate}[noitemsep]
		\item[(1)] $\mx$ is a tree, has an $\mx$-branching node, and is good.
		\item[(2)] $L_i \leq \adm(\mx)\leq 20L_i$.
		\item[(3)] $|\mv(\mx)| = \Omega(\frac{1}{\epsilon})$  when $\epsilon \leq 2/g$. 
		\item[(4)] Let $\Ftilde^{(3)}_i$ be obtained from $\Ftilde^{(2)}_i$ by removing every node contained in subgraphs of $\mathbb{X}_2$ and in $\mz$. Then, for every tree $\Ttilde \subseteq \Ftilde^{(3)}_i$, either (4a) $\adm(\Ttilde)\leq 6L_i$ or (4b) $\Ttilde$ is a path.
		\item[(5)] Nodes in $\mz$ are augmented to subgraphs in $\mathbb{X}_1$ such that for every subgraph $\my \in \mathbb{X}_1$ that are augmented, $\my^{\aug}$ remains a tree and $\adm(\my^{\aug}) \leq 24L_i$ where $\my^{\aug}$ is $\my$ after the augmentation.  
	\end{enumerate}
\end{restatable}

There are two differences in the construction of Step 2 in our construction compared to the construction in LS algorithm. First, the graphs constructed are good. Second, for some edges cases where we could not group branching nodes into subgraphs satisfying Item (1), we show that they could be augmented to subgraphs in $\mathbb{X}_1$. These nodes are in  $\mz$ in Item (5), and our construction guarantees that the augmentation does not change the structure of subgraphs in $\mathbb{X}_1$. That is, subgraphs in $\mathbb{X}$ remain trees, and their diameters are not increased by much. The increase in the diameter from $13 L_i$ in \Cref{lm:Clustering-Step1} to $24L_i$ in Item (5) in \Cref{lm:Clustering-Step2} does not affect the overall argument of Le and Solomon~\cite{LS21}; this only affects the choice of $g$, which we have the freedom to choose as large as we want.  The augmented diameter of $\mx$ in Item (2) in \Cref{lm:Clustering-Step2} is also slightly larger than the diameter of subgraphs in~\cite{LS21}, which is at most $2L_i$. This change also only affects the choice of $g$. The proof of \Cref{lm:Clustering-Step2} will be delayed to \Cref{subsec:proof}.

\paragraph{Step 3: Augmenting $\mathbb{X}_1\cup \mathbb{X}_2$.~} We say that a path of augmented diameter at least $6L$ in the forest $\Ftilde^{(3)}_i$ in Item (4) of \Cref{lm:Clustering-Step2} a \emph{long path}.  In this step, we further augment graphs formed in Steps 1 and 2. The purpose is to guarantee that for any long path after this step, at least one endpoint of the path is connected to a node in a subgraph of $\mathbb{X}_1\cup \mathbb{X}_2$ via an $\msttilde_{i}$ edge. 
\begin{quote}
\textbf{The construction.~} Let $\mathcal{A}$ be the set of all nodes in a long path of   $\Ftilde^{(3)}_i$  that is $\msttilde_{i}$-branching.  For each node $\varphi \in \mathcal{A}$, let $\mx \in \mathbb{X}_1\cup \mathbb{X}_2$ be (any) subgraph such that $\varphi$ is connected to a node in $\mx$ via an $\msttilde_{i}$ edge $\mbe$.  We then add $\varphi$ and $\mbe$ to $\mx$.  
\end{quote}

\begin{lemma}[Lemma 5.3.~\cite{LS21}]\label{lm:Clustering-Step3} The augmentation in Step 3 can be implemented in $O(|\mv_i|)$ time, and increases the augmented diameter of each subgraph in  $\mathbb{X}_1\cup \mathbb{X}_2$ by at most $4L_i$ when $\eps \leq 1/g$. \\
	Furthermore, let $\Ftilde^{(4)}_i$ be the forest obtained from $\Ftilde^{(3)}_i$ by removing every node in $\mathcal{A}$. Then, for every tree $\Ttilde \subseteq \Ftilde^{(4)}_i$, either:
	\begin{enumerate}[noitemsep]
		\item[(1)]$\adm(\Ttilde)\leq 6L_i$ or
		\item[(2)] $\Ttilde$ is a path such that (2a)  every node in $\Ttilde$ has \emph{degree at most $2$} in $\msttilde_{i}$ and (2b) at least one endpoint $\varphi$ of  $\Ttilde$ is connected via an $\msttilde_{i}$ edge to a node $\varphi'$ in a subgraph of $\mathbb{X}_1\cup \mathbb{X}_2$, unless $\mathbb{X}_1\cup \mathbb{X}_2 = \emptyset$. 
	\end{enumerate}
\end{lemma}

We emphasize that in Item (2a) of \Cref{lm:Clustering-Step3}, the degree bound is in $\msttilde_{i}$. This is important for the construction in Step 5. Step 4 deals with long paths of $\Ftilde^{(4)}_i$, the forest in \Cref{lm:Clustering-Step3}. The construction uses Red/Blue Coloring. The coloring guarantees that for any long path in $\Ftilde^{(4)}_i$, the nodes in the prefix/suffix of augmented length at most $L_i$ get red color, while other nodes get blue color.  

\begin{quote}
	\textbf{Red/Blue Coloring.~}\hypertarget{RBColoring}{}  The coloring applies to each long path $\Ptilde \in  \Ftilde^{(4)}_i$. Specifically, a node gets red color if its augmented distance to at least one of the two endpoints of $\Ptilde$ is at most $L_i$; otherwise, it gets blue color. 
\end{quote}

\begin{lemma}[Step 4, Lemma 5.4~\cite{LS21}]\label{lm:Clustering-Step4} We can construct in $O((|\mv_i| + |\me_i|)\epsilon^{-1})$ time a collection $\mathbb{X}_4$ of subgraphs of $\mg_i$ such that every $\mx\in \mathbb{X}_4$:
	\begin{enumerate}[noitemsep]
		\item[(1)] $\mx$ contains a single edge in $\me_i$.
		\item[(2)] $L_i \leq \adm(\mx)\leq 5L_i$.
		\item[(3)]  $|\mv(\mx)| = \Theta(\frac{1}{\epsilon})$ when $\epsilon \ll \frac{1}{g}$. 
		\item[(4)] $\Delta_{i+1}^{+}(\mx) = \Omega(\eps^2 |\mv(\mx)| L_i)$.
		\item[(5)] Let $\Ftilde^{(5)}_i$ be obtained from $\Ftilde^{(4)}_i$ by removing every node contained in subgraphs of $\mathbb{X}_4$. If we apply \hyperlink{RBColoring}{Red/Blue Coloring} to each path of augmented diameter at least $6L_i$ in $\Ftilde^{(5)}_i$, then there is no edge in $\me_i$ that connects two blue nodes in $\Ftilde^{(5)}_i$.
	\end{enumerate}
\end{lemma}

Item (5)  of \Cref{lm:Clustering-Step4} guarantees that for any edge with one endpoint in a long path of $\Ftilde^{(5)}_i$, at least one of the endpoints must have red color. $\Ftilde^{(5)}_i$ has the following structure.

\begin{observation}[Observation 5.7~\cite{LS21}]\label{obs:Clustering-F5} Every tree $\Ttilde \subseteq \Ftilde^{(5)}_i$ of augmented diameter at least $6L_i$ is connected via $\msttilde_{i}$ edge to a node in some subgraph $\mx \in \mathbb{X}_1 \cup \mathbb{X}_2\cup \mathbb{X}_4$, unless there is no subgraph formed in Steps 1-4, i.e., $ \mathbb{X}_1 \cup \mathbb{X}_2\cup \mathbb{X}_4 = \emptyset$.
\end{observation}

We observe that any tree $\Ttilde \subseteq \Ftilde^{(5)}_i$  of diameter at least $6L_i$ must be a path, and that, by Item (2a) in \Cref{lm:Clustering-Step3}, only endpoints of $\Ttilde$ could have an edge in $\msttilde_{i}$ to a node outside $\Ttilde$. We call such an endpoint a \emph{connecting endpoint} of $\Ttilde$. Note that $\Ttilde$ could have up to two connecting endpoints. 

Step 5 has two smaller steps. In Step 5A, we augment trees of $\Ftilde^{(5)}_i$ of low augmented diameter to existing subgraphs. In Step 5B, we form new subgraphs from long paths, and augment the prefix/suffix to an existing subgraph in previous steps.

\paragraph{Step 5.~}  Let $\Ttilde$ be  a path in  $\Ftilde^{(5)}_i$ obtained by Item (5) of \Cref{lm:Clustering-Step4}. We construct two sets of subgraphs, denoted by $\mathbb{X}^{\internal}_5$ and $\mathbb{X}^{\prefix}_5$.
\begin{itemize}
	\item (Step 5A)\hypertarget{5A}{}  If $\Ttilde$ has augmented diameter at most $6L_i$, let $\mbe$ be an $\widetilde{\mst}_i$ edge connecting $\Ttilde$  and a node in some subgraph $\mx \in \mathbb{X}_1\cup \mathbb{X}_2 \cup \mathbb{X}_4$, assuming that $\mathbb{X}_1\cup \mathbb{X}_2 \cup \mathbb{X}_4 \not= \emptyset$. We add both $\mbe$ and $\Ttilde$ to $\mx$.
	\item (Step 5B)\hypertarget{5B}{} 	Otherwise, $\Ttilde$  is a path. We break $\Ttilde$ into subpaths of augmented diameter at least $L_i$ and at most $7L_i$ by applying the construction in \Cref{lm:Clustering-Step5B} below. For any subpath $\Ptilde$ broken from $\Ttilde$, if $\Ptilde$ is connected to a node in a subgraph $\mx$ via an  edge $\mbe\in \msttilde_{i}$, we add $\Ptilde$ and $\mbe$ to $\mx$; 	otherwise,  $\Ptilde$ becomes a new subgraph.  We add $\Ptilde$ to $\mathbb{X}^{\prefix}_5$ if it is a prefix/suffix of $\Ttilde$; otherwise, we add $\Ptilde$ to $\mathbb{X}^{\internal}_5$. 
\end{itemize}

\begin{figure}[!h]
	\begin{center}
		\includegraphics[width=1.0\textwidth]{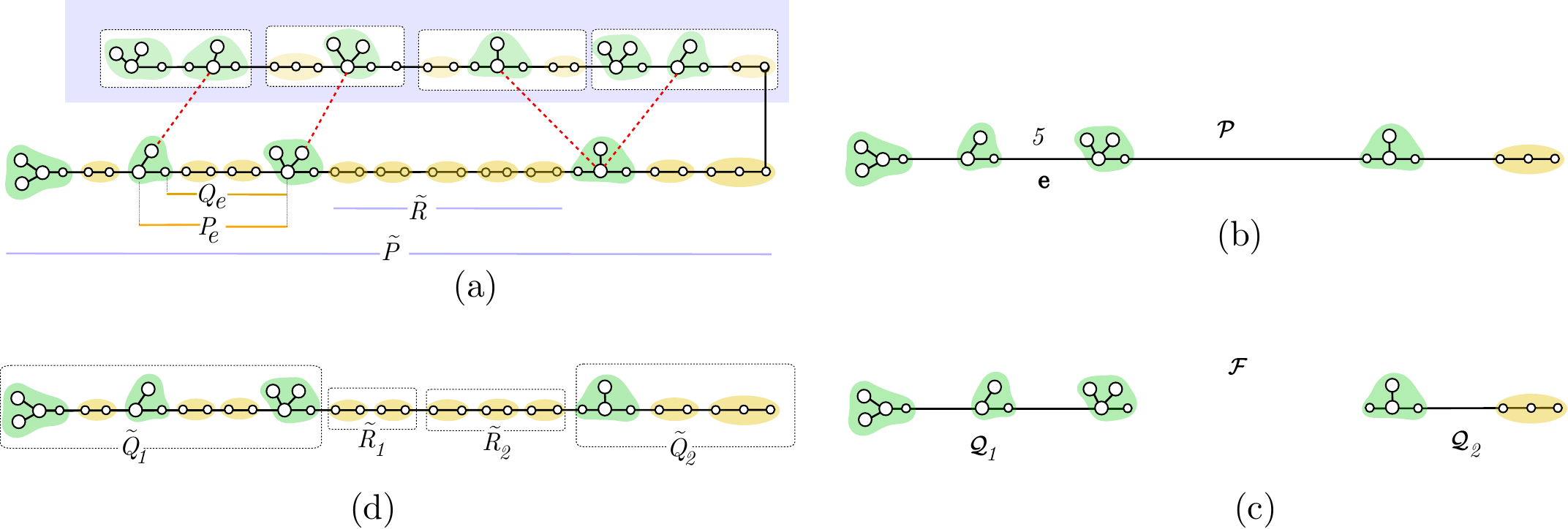}
	\end{center}
	\caption{An illustration for the proof of \Cref{lm:Clustering-Step5B}. Small circles are virtual vertices; black (solid) edges are $\msttilde$ edges and red (dashed) edges are edges in $\me_i$. (a) Non-isolated nodes in $\Ptilde$ are those incident to red edges.   Nodes grouped in previous steps are in the blue-shaded region. The path $Q_e$ corresponds to an $\mbe$ in $\mathcal{P}$ is highlighted. $Q_e$ is a subpath of the parent path $P_e$ of the virtual clusters in the construction of $\mbe$. (b) The path $\mathcal{P}$ obtained from $\Ptilde$ in figure (a) by the construction in the proof  of \Cref{lm:Clustering-Step5B}; the only virtual node in $\mathcal{P}$ is the (connecting) endpoint of $\mathcal{P}$. Suppose that every edge in $Q_e$ in figure (a) has weight $1$, then $\mbe$ has weight $5$ since $Q_e$ has 5 edges. In general, $\omega(\mbe) = w(Q_e)$. (c) Forest $\mathcal{F}$ obtained from $\mathcal{P}$ by removing every edge of weight at least $2L_i$. $\mathbb{A}$ in this case includes two paths $\mathcal{Q}_1$ and $\mathcal{Q}_2$. (d) Two paths $\tilde{Q}_1$ and $\tilde{Q}_2$ in $\mathbb{P}$ constructed from $\mathcal{Q}_1$ and $\mathcal{Q}_2$, respectively. Two other paths $\tilde{R}_1$ and $\tilde{R}_2$ are broken from the path $\tilde{R}$ in (a).}
	\label{fig:Step5B}
\end{figure}

\begin{restatable}{lemma}{ClusteringLastStep}
	\label{lm:Clustering-Step5B}  Let  $\Ptilde$  be a path of augmented diameter at least $6L_i$  in  $\Ftilde^{(5)}_i$.  We can break $\Ptilde$ into a collection of paths $\mathbb{P}$ such that each path $\Ptilde' \in \mathbb{P}$ has two properties:
	\begin{enumerate}[noitemsep]
		\item[(1)] $L_i \leq \adm(\Ptilde')\leq 7L_i$.
		\item[(2)] If $\Ptilde'$ contains a non-isolated node, then it contains at least two non-virtual nodes, or a connecting endpoint of $\Ptilde$.
	\end{enumerate}
	The running time of the construction is $O(|\mv(\Ptilde)|)$. 
\end{restatable}
\begin{proof}
	Recall that by Item (2a) in \Cref{lm:Clustering-Step3}, every node in $\Ptilde$ has degree 2 in $\msttilde_{i}$. This means, if an endpoint of $\Ptilde$ is non-connecting, then it is a non-virtual node. Recall by the definition of a virtual node $\varphi_C$, its corresponding cluster $C$ is virtual, and hence, structurally, $C$ induced a subpath of the parent path $P_e$.
	
	We construct path graph $\mathcal{P}$ from $\Ptilde$ that contains non-virtual nodes and the endpoints of $\Ptilde$ as follows.	Each edge $\mbe = (\varphi,\varphi') \in \mathcal{P}$ corresponds to a path between $\varphi$ and  $\varphi'$ in $\Ttilde$ whose internal nodes are virtual. Note that all virtual nodes on the path between $\varphi$ and  $\varphi'$ in $\Ttilde$ share the same parent path $P_e$. Let $Q_e$ be the minimal subpath of $P_e$ whose endpoints are in the clusters corresponding to $\varphi$ and $\varphi'$. We then assign a weight $\omega(\mbe) = w(Q_e)$. Observe that $\omega(\mbe)\leq w(P_e) = w(e)$ where $e$ is the $\mst$ edge from which $P_e$ is subdivided. See \Cref{fig:Step5B}(a) and (b) for an illustration.
	
	Note by Item (2a) of \Cref{lm:Clustering-Step2}, every node in $\Ptilde$ has degree at most $2$ in $\msttilde_{i}$. If $\varphi$ is a non-isolated node in $\Ptilde$, then it is incident to an edge, say $\mbe'$, in $\me_i$ by definition. One of the incident edges of $\varphi$ is part of the fundamental cycle of $\msttilde_{i}$ formed by $\mbe'$.  It follows from \Cref{lm:cycle-Prop} that at least one edge in $\mathcal{P}$ of $\varphi$ must have a weight at most $L_i$.

	Let $\mathcal{F}$ be the forest induced by edges of weight at most $2L_i$ in $\mathcal{P}$. We further remove singletons from $\mathcal{F}$. Observe that a singleton in $\mathcal{F}$ is either a connecting endpoint of $\mathcal{P}$, or an isolated node.  We then greedily break each path in $\mathcal{F}$ that contains at least three edges into subpaths of at least two edges and at most three edges each. As a result, we obtain a collection $\mathbb{A}$ of subpaths of $\mathcal{P}$ that contain at least two nodes each. See \Cref{fig:Step5B}(c). 
	
	We now construct $\mathbb{P}$ as follows.  (Step 1) For each path $\mathcal{Q}\in \mathbb{A}$, we construct the corresponding subpath $\Qtilde$ of $\Ptilde$ by replacing each edge in $\mathcal{Q}$ by the corresponding subpath in $\Ptilde$. We then add $\Qtilde$ to $\mathbb{P}$.  (Step 2) After  Step 1, remaining nodes in $\Ptilde$ that are not grouped to a path in $\mathbb{P}$ induces a collection of subpaths, say $\mathbb{Q}$, of $\Ptilde$. Observe by the construction of $\mathcal{F}$ that,  each subpath in the collection $\mathbb{Q}$ corresponds to a subpath of $\mathcal{P}$, which only contains virtual nodes and isolated nodes, that has at least one edge  of weight at least $2L_i$. Now for each path $\tilde{R} \in \mathbb{Q}$, observe that $\adm(\tilde{R})\geq 2L_i  - 2\bar{w} - 2g\eps L_i \geq 2L_i - 4g\eps L_i\geq L_i$ when $\eps \leq 1/2g$. The negative term $ - 2\bar{w} - 2g\eps L_i$ is due to that the two nodes neighboring the endpoints of $\tilde{R}$ are grouped to subpaths in $\mathbb{P}$. We then break $\tilde{R}$ into subpaths of augmented diameter at least $L_i$ and at most $2L_i$ and add them to $\mathbb{P}$. This completes the construction of $\mathbb{P}$. See Figure~\ref{fig:Step5B}(d) for an illustration.
	
	The running time follows directly from the construction. To bound the augmented diameter of paths in $\mathbb{P}$, we observe that  path $\Qtilde$ in Step 1 has augmented diameter at most $3(2L_i) + 4\eps g L_i \leq 7L_i$ when $\eps \leq 1/4g$. The additive term $4\eps g L_i$ is due to (at most) four endpoints of (at most) three edges in $\Qtilde$. Thus, every path in $\mathbb{P}$ has an augmented diameter of at most $\max\{7L_i,2L_i\} = 7L_i$. The lower bound $L_i$ follows directly from the construction; this implies Item (1). Item (2) follows from the construction of $\mathbb{A}$. \qed
\end{proof}

We note that in Step 5B in LS algorithm,  $\Ttilde$ is broken into subpaths of augmented length at least $L_i$ and at most $2Li$ instead of at least $L_i$ and at most $7L_i$ as in our construction.  The increase in the augmented diameter ultimately affects the choice of $g$.  Other properties of subgraphs in  $ \mathbb{X}_5^{\internal}$ and $\mathbb{X}_5^{\prefix}$  remains the same.

\begin{lemma}[Lemma 5.8~\cite{LS21}]\label{lm:Clustering-Step5} We can implement the construction of $ \mathbb{X}_5^{\internal}$ and $\mathbb{X}_5^{\prefix}$ in $O(|\mv_i|)$ time.	Furthermore, every subgraph $\mx \in \mathbb{X}_5^{\internal} \cup \mathbb{X}_5^{\prefix}$ satisfies:
	\begin{enumerate}[noitemsep]
		\item[(1)] $\mx$ is a subpath of $\msttilde_{i}$.
		\item[(2)] $L_i \leq \adm(\mx)\leq 7 L_i$.
		\item[(3)] $|\mv(\mx)| = \Theta(\frac{1}{\epsilon})$.
	\end{enumerate}
\end{lemma}

We note that the degenerate case in the above construction happens when  $\mathbb{X}_1 \cup \mathbb{X}_2\cup \mathbb{X}_4 = \emptyset$. When the degenerate case happens, $\Ftilde^{(5)}_i$ has the following structure. 

\begin{lemma}[Lemma 5.10~\cite{LS21}]\label{lm:exception}
	If  $\mathbb{X}_1\cup \mathbb{X}_2\cup \mathbb{X}_4 = \emptyset$, then
	$\Ftilde^{(5)}_i =  \msttilde_{i}$, and $\msttilde_{i}$  is a single (long) path.   Moreover, every edge $\mbe \in \me_i$ must be incident to a  node in $\Ptilde_1\cup \Ptilde_2$,
	where $\Ptilde_1$ and $\Ptilde_2$ are the prefix and suffix subpaths of $\msttilde_{i}$ of augmented diameter at most $L_i$. Furthermore, $|\me_i| = O(1/\epsilon^2)$.
\end{lemma}

We are now ready to prove \Cref{lm:additional-prop}.

\paragraph{Proof of \Cref{lm:additional-prop}.~} The degree-specific partition $\mathbb{V}$ of $\mv_i$ and the partition of $\mathbb{X}$ conforming $\mathbb{V}$ are constructed as follows. If the degenerate case happens, then  $\mv^{\lowm}_i = \mv_i$ (and hence $\mv^{\high}_i = \mv^{\lowp}_i = \emptyset$). In this case, $\mathbb{X}^{\lowm} =  \mathbb{X}^{\internal}_5\cup \mathbb{X}^{\prefix}_5$, while $\mathbb{X}^{\high} = \mathbb{X}^{\lowp} = \emptyset$.  Otherwise, 
$\mv^{\high}_i$ to be the set of all nodes that are incident to at least $2g/\eps$ edges in $\me_i$ in \Cref{lm:Clustering-Step1},  $\mv^{\lowm}_i = \cup_{\mx \in \mathbb{X}^{\internal}_5}\mv(\mx)$ and $\mv^{\lowp}_i = \mv_i\setminus (\mv^{\high}_i \cup \mv^{\lowm}_i )$. The partition of $\mathbb{X}$ is $\{\mathbb{X}^{\high} =\mathbb{X}_1$, $\mathbb{X}^{\lowp} = \mathbb{X}_2\cup \mathbb{X}_4 \cup \mathbb{X}^{\prefix}_5$,$\mathbb{X}^{\lowm} = \mathbb{X}^{\internal}_5\}$.

We note that Items (1) and (2) in \Cref{lm:Clustering} hold by the same proof in~\cite{LS21}. For Item (3), subgraphs in $\mathbb{X}$ satisfy all properties (\hyperlink{P1'L}{P1'})-(\hyperlink{P3'L}{P3'}) with constant $g = 42$ instead of $31$ since the construction of Step 2 in \Cref{lm:Clustering-Step2} increases the augmented diameter of subgraphs in $\mathbb{X}_1$ by $11L_i$ (on top of the upper bound $31L_i$). We remark that the augmented diameter of other subgraphs is smaller than the augmented diameters of subgraphs in $\mathbb{X}_1$, and hence, the increased diameter due to our construction does not affect $g$. The fact that $|\me(\mx)\cap \me_i| = O(|\ma_{\mx}|)$ where $\ma_{\mx}$ is the set of nodes in $\mx$ incident to an edge in $\me(\mx)\cap \me_i$ follows from that $\mx$ is a tree for all cases, except in Step 4 (\Cref{lm:Clustering-Step4}). However, in this case, $\mx$ has a single edge in $\me_i$, and hence $|\me(\mx)\cap \me_i| \leq 1 =  O(|\ma_{\mx}|)$. 

It remains to show the reduction in the number of non-virtual clusters as claimed in \Cref{lm:additional-prop}. All we need to show is that for every subgraph $\mx$ that contains a non-isolated node, it contains at least two non-virtual nodes. That is, $\mx$ is good. This holds for subgraphs in $\mathbb{X}_1\cup \mathbb{X}_4$, since every subgraph in this set contains at least one edge in $\me_i$, whose endpoints are non-isolated by the definition of a non-isolated node.  Every subgraph in $\mathbb{X}_2$ is good by Item (1) in \Cref{lm:Clustering-Step2}. Observe that each subgraph  $\mx \in \mathbb{X}^{\internal}_5\cup \mathbb{X}^{\prefix}_5$ corresponds to a subpath of $\Ttilde$ in Step 5B that does not contain the connecting endpoint. By Item (2) in \Cref{lm:Clustering-Step5B}, $\mx$ contains at least two non-virtual nodes, if it contains at least one non-isolated node, and hence $\mx$ is good.  \Cref{lm:additional-prop} now follows. \qed

\subsubsection{Proof of \Cref{lm:Clustering-Step2}}\label{subsec:proof}
In this section, we provide the proofs of \Cref{lm:Clustering-Step2}, which we restate below.

\begin{figure}[h]
	\begin{center}
		\includegraphics[width=0.8\textwidth]{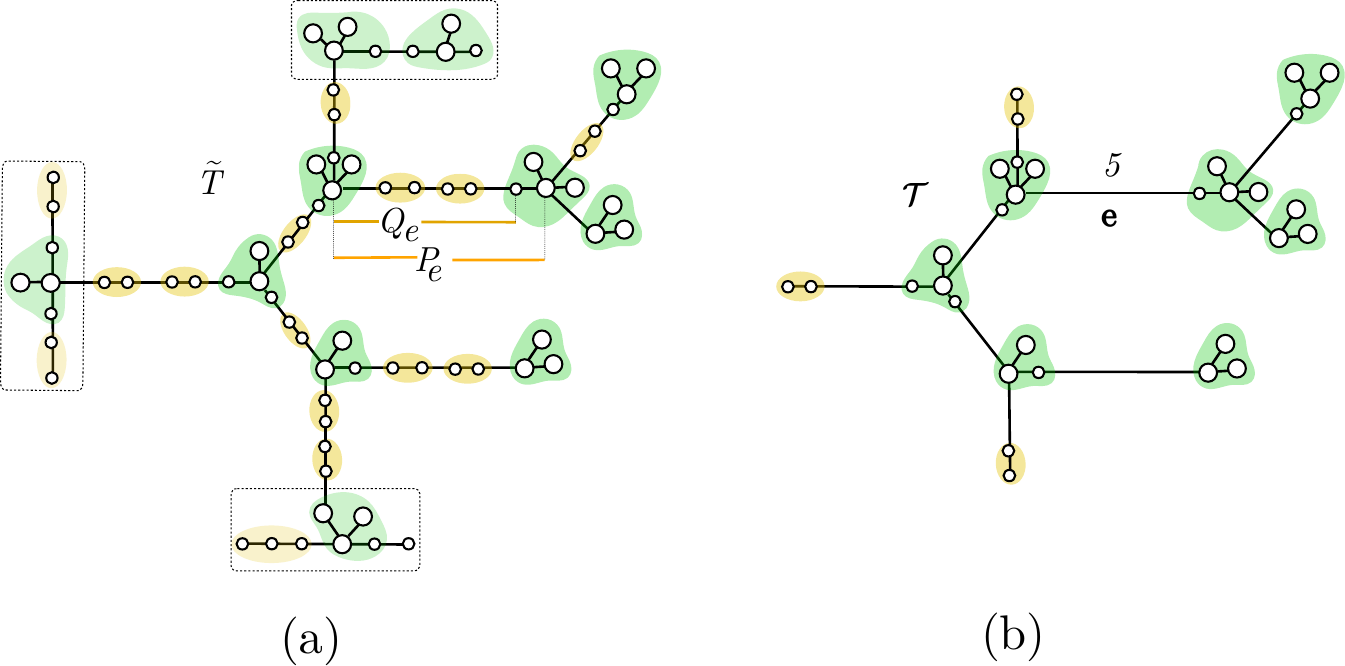}
	\end{center}
	\caption{Virtual clusters are yellow shaded and non-virtual clusters are green shaded. Virtual vertices are small circles. (a) A tree $\Ttilde$ considered in the construction of Step 2 (\Cref{lm:Clustering-Step2}). (b) The tree $\mathcal{T}$ constructed from non-virtual nodes and connecting nodes of $\Ttilde$ in the proof of \Cref{lm:Clustering-Step2}. If every edge in the path $Q_e$ has weight $1$ as in figure (a), then the weight of $\mbe$ in figure (b) is $5$. In general, $\omega(\mbe)  = w(Q_e)$. Every virtual node in $\mathcal{T}$ is a connecting node. Subgraphs in the rectangular dashed curves are subgraphs formed in previous steps.}
	\label{fig:Step2}
\end{figure}

\ClusteringStepTwo*
\begin{proof}
Let $\Ttilde$ be a tree of augmented diameter at least $6L_i$ in $\Ftilde^{(2)}$. We say that a node $\varphi \in \Ttilde$ is a connecting node if it has an $\mst$ edge to a subgraph $\mx \in \mathbb{X}_1$.

We now construct a tree $\mathcal{T}$ in the same way we construct a path $\mathcal{P}$ in \Cref{lm:Clustering-Step5B}. $\mathcal{T}$ is a tree that contains non-virtual nodes and connecting nodes of $\Ttilde$, which may or may not be virtual. Note that branching nodes of $\Ttilde$ are non-virtual. Each edge $\mbe = (\varphi,\varphi') \in \mathcal{T}$ corresponds to a path between $\varphi$ and  $\varphi'$ in $\Ttilde$ whose internal nodes are virtual. Note that all virtual nodes on the path between $\varphi$ and  $\varphi'$ in $\Ttilde$ share the same parent path $P_e$. Let $Q_e$ be the minimal subpath of $P_e$ whose endpoints are in the clusters corresponding to $\varphi$ and $\varphi'$. We then assign a weight $\omega(\mbe) = w(Q_e)$. Observe that $\omega(\mbe)\leq w(P_e) = w(e)$ where $e$ is the $\mst$ edge from which $P_e$ is subdivided. See Figure~\ref{fig:Step2} for an illustration.

\begin{claim}\label{clm:T-Prop} If a node $\varphi$ in $\Ttilde$  is non-isolated and non-connecting, then $\varphi$ is incident to an edge of weight at most $L_i$ in $\mathcal{T}$.
\end{claim}
\begin{proof}
	By definition of a non-isolated node, $\varphi$ is incident to an edge, say $\mbe'$, in $\me_i$ by definition. One of the incident edges of $\varphi$ belongs the fundamental cycle of $\msttilde_{i}$ formed  by $\mbe'$.  It follows from \Cref{lm:cycle-Prop} that at least one edge in $\mathcal{T}$ of $\varphi$ must have a weight at most $\omega(\mbe')\leq L_i$. \qed
\end{proof}

We first apply the following construction to obtain a collection of trees, say $\mathbb{A}$, and then we will post-process the trees to obtain $\mathbb{X}_2$ as claimed in \Cref{lm:Clustering-Step2}.  We say that a tree $\Ttilde$ in $\Ftilde^{(2)}$ a long tree if its augmented diameter is at least $6L_i$. The construction of $\mathbb{A}$ is similar to Step 2 in LS algorithm, except that the radius of the BFS step in our construction is slightly larger. 

\begin{itemize}
	\item (Step i) Pick a long tree  $\Ttilde$ of $\Ftilde^{(2)}_i$ with at least one $\Ttilde$-branching node, say $\varphi$. If $\Ttilde$ has a $\Ttilde$-branching node that is non-isolated, we then choose $\varphi$ to be a non-isolated node.  We traverse $\Ttilde$ by BFS starting from $\varphi$ and {\em truncate} the traversal at nodes whose augmented distance from $\varphi$ is at least $2L_i$. The augmented radius (with respect to the center $\varphi$) of the subtree induced by the visited nodes is at least $L_i$ and at most $2L_i + \bar{w} + g\epsilon L_i \leq 2L_i + 2g\eps L_i$.  We then create a new tree $\Ttilde'$ induced by the visited nodes. 
\end{itemize}

After the construction in Step i, every tree in $\Ttilde$ either has augmented diameter at most $6L_i$ or is a path. 

An important property that we would like to have is that every tree in $\mathbb{A}$ either contains no non-isolated node or at least two non-virtual nodes. To this end, we need to post-process $\mathbb{A}$. Our postprocessing relies on the following structure of trees in $\mathbb{A}$.

\begin{figure}[h]
	\begin{center}
		\includegraphics[width=0.8\textwidth]{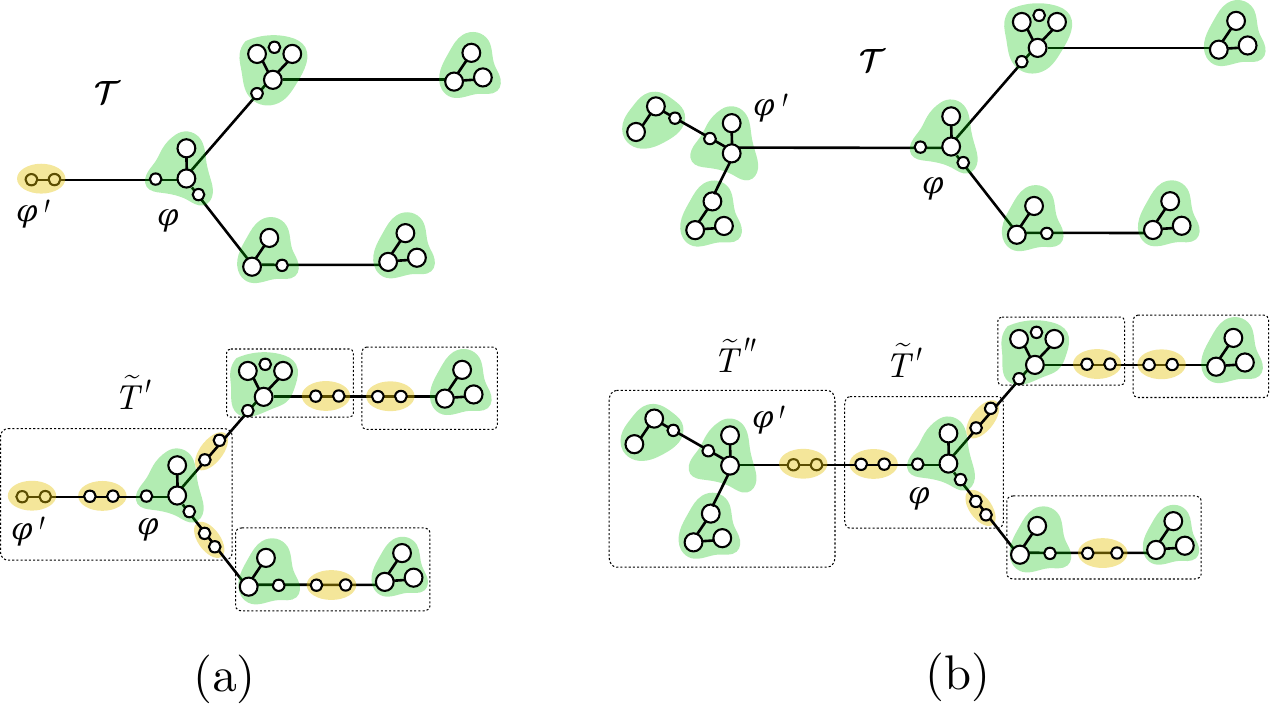}
	\end{center}
	\caption{Two cases in the proof of \Cref{clm:adj}. (a) $\varphi'$ is a virtual node in $\mathcal{T}$. Then it is a connecting node, and is grouped to $\Ttilde'$. (b) $\varphi'$ is a non-virtual node. Then it is grouped in $\Ttilde''$ that is adjacent to $\Ttilde'$. $\Ttilde''$ contains at least two non-virtual nodes (3 non-virtual nodes in this figure).}
	\label{fig:Step2-Proof} 
\end{figure}

\begin{claim}\label{clm:adj}  Let $\Ttilde' \in \mathbb{A}$ be a tree that contains exactly one non-isolated node, no connecting node, and no other non-virtual node. Then $\Ttilde'$ is adjacent to a tree $\mathcal{T}''\in \mathbb{A}$ that has at least two non-virtual nodes. 
\end{claim}
\begin{proof}Let $\varphi$ be the non-isolated node in $\Ttilde'$. Observe that the center of $\Ttilde'$ is a branching node, and hence, is non-virtual. It follows that $\varphi$ must be the center of $\Ttilde'$ since otherwise, $\Ttilde'$ contains two non-virtual nodes, contradicting the assumption of the claim. Let $\varphi'$ be the neighbor in $\mathcal{T}$ of $\varphi$ whose edge $(\varphi',\varphi)$ has weight at most $L_i$ by \Cref{clm:T-Prop}. By construction, the radius of the traversal is at least $2L_i > L_i + 2\eps gL_i $ when $\eps \leq 1/4g$. If $\varphi'$ is a virtual node (see Figure~\ref{fig:Step2-Proof}(a)), then it must be connecting, and hence $\varphi'$ belong to $\Ttilde'$, contradicting that $\Ttilde'$ has no connecting node. Otherwise, $\varphi'$ is a non-virtual node and is grouped into another tree, say $\Ttilde''\in \mathbb{A}$ (see Figure~\ref{fig:Step2-Proof}(b)). Observe that $\Ttilde'$ and $\Ttilde''$ are adjacent, i.e., connected by an edge in $\msttilde_{i}$, since all nodes between $\varphi$ and $\varphi'$ have degree 2 as they are virtual nodes. We claim that $\Ttilde''$  must have at least two non-virtual nodes.  If $\varphi'$ is not a center of $\Ttilde''$, then $\Ttilde''$ contains at least two non-virtual nodes since its center is a non-virtual node. Otherwise, $\varphi'$ is the center of $\Ttilde''$, and hence, $\varphi$ would have been merged to $\Ttilde''$ during the construction of $\Ttilde''$, a contradiction. 	\qed
\end{proof}

\noindent Our construction in the next step is as follows.

\begin{itemize}
	\item (Step ii) Pick a tree $\Ttilde'$ in $\mathbb{A}$ that has one non-isolated node and no other non-virtual node. If $\Ttilde'$ contains a connecting node, say $\varphi$. Let $\my\in \mathbb{X}_1$ be a subgraph such that $\varphi$ has an $\msttilde_{i}$ edge $\mbe$ to a node in $\my$. We then add $\Ttilde'$ and $\mbe$ to $\my$, and add the set of nodes of $\Ttilde'$ to $\mathcal{Z}$. Otherwise, $\Ttilde'$ is adjacent to another tree $\Ttilde''\in \mathbb{A}$ that has at least two non-virtual nodes by \Cref{clm:adj}. We then add $\Ttilde'$ and the $\msttilde_{i}$ edge connecting $\Ttilde'$ and $\Ttilde''$ to $\Ttilde''$. We then repeat this step until it no longer applies. The set $\mathbb{X}_2$ is the set of trees in $\mathbb{A}$ after this step completed.
\end{itemize}

We now prove all properties in \Cref{lm:Clustering-Step2}. Step i is the same as Step 2 in LS algorithm and hence can be implemented in $O(\mv_i)$ following~\cite{LS21} (Lemma 5.2). Step ii can be implemented in $O(|\mv(\Ftilde^{(2)})|) = O(\mv_i)$ by following each step of the construction. Thus, the total running time is $O(|\mv_i|)$.

Item (1) of \Cref{lm:Clustering-Step2} and Item (4) follows directly from the construction.  By the construction in Step i, every tree has an augmented diameter at least $2L_i$ and at most $(2L_i + 2\eps L_i )$. The augmentation in Step ii is done via a star-like way, and hence, increases the diameter of each tree in $\mathbb{A}$ by at most $2(2L_i + 2\eps g L_i) + 2\bar{w} \leq 2(2L_i + 2\eps g L_i) + 2g\eps L_i = 4L_i + 6\eps L_i$. (Here we use the fact that $\bar{w}\leq L_{i-1} = \eps L_i \leq g\eps L_i$.) Thus, the final diameter is at most $2L_i + 2\eps g L_i +  4L_i + 6\eps L_i \leq 6L_i + 8\eps L_i \leq 14L_i < 20L_i $ when $\eps \leq 1/g$; this implies Item (2) of \Cref{lm:Clustering-Step2}.

For Item (3), note that each tree $\mx \in \mathbb{X}_2$ has augmented diameter at least $2L_i$, and that  every edge/node has a weight at most $\max\{\bar{w}, g\eps L_i\} = g\eps L_i$. It follows that  $|\mv(\mx)|\geq \frac{2L_i}{g\eps L_i} = \Omega(1/\eps)$, as claimed.

For Item (5), we observe that each subgraph $\my\in \mathbb{X}_1$ is augmented in Step ii via an $\msttilde_i$ edges an in a star-like way. Thus, $\adm(\my^{\aug})\leq \adm(\my) + 2(2L_i + 2\eps g L_i) + 2\bar{w} \leq \adm(\my)+  4L_i + 6\eps g L_i \leq 13 L_i+  4L_i + 6g \eps L_i \leq 23L_i < 24L_i$ when $\eps \leq 1/g$. This complete the proof of \Cref{lm:Clustering-Step2}.\qed
\end{proof}

\paragraph*{Acknowledgments.}  Hung Le is supported by a start up funding of University of Massachusetts at Amherst and the  National Science Foundation under Grant No. CCF-2121952. Shay Solomon is partially supported by the Israel Science Foundation grant No.1991/19 and by Len Blavatnik and the Blavatnik Family foundation.
\bibliographystyle{alpha}
\bibliography{spanner}
\pagebreak
\appendix

\end{document}